\documentclass[runningheads]{llncs}
\usepackage{algorithm,algpseudocode}
\usepackage{mathtools}

\usepackage{graphicx}
\usepackage{tabularx}
\usepackage{enumitem}
\usepackage{textcomp}
\usepackage{xcolor}
\usepackage{enumitem}
\usepackage{float}
\usepackage[font=scriptsize]{caption}
\usepackage[labelformat=simple]{subcaption}

\usepackage{booktabs}

\usepackage{pgf,tikz,pgfplots}
\pgfplotsset{compat=1.15}
\usepackage{mathrsfs}
\usepackage{hyperref}
\usepackage{multicol}
\usepackage{multirow}
\usepackage{amssymb}
\usepackage{xcolor}
\usepackage{algorithm}
\usepackage{algpseudocode}
\usepackage[export]{adjustbox}
\newcommand{\cC}{{\mathcal C}}
\newcommand{\cW}{\mathcal W}
\newcommand{\tL}{Tr(G)}

\newcommand{\maxG}{max_{6}(\mathcal{I}_{SP}(G))}
\usepackage{orcidlink}

\begin{document}
\title{Graph Traversal via Connected Mobile Agents}

\author{Saswata Jana\inst{1}\orcidlink{0000-0003-3238-8233} \thanks{Supported by  Prime Minister's Research Fellowship (PMRF) scheme of the Govt. of India (PMRF-ID: 1902165)} \and Giuseppe F. Italiano\inst{2}\orcidlink{0000-0002-9492-9894} \and Partha Sarathi Mandal\inst{1}\orcidlink{0000-0002-8632-5767}}
\authorrunning{Jana et al.}
\institute{Indian Institute of Technology Guwahati, Guwahati, India \and Luiss University, Rome, Italy}
\email{}

\maketitle
\begin{abstract}
This paper considers the Hamiltonian walk problem in the multi-agent coordination framework, referred to as $k$-agents Hamiltonian walk problem ($k$-HWP).
In this problem, a set of $k$ connected agents collectively compute a spanning walk of a given undirected graph in the minimum steps.
At each step, the agents are at $k$ distinct vertices and the induced subgraph made by the occupied vertices remains connected. 
In the next consecutive steps, each agent may remain stationary or move to one of its neighbours.
To the best of our knowledge, this problem has not been previously explored in the context of multi-agent systems with connectivity.
As a generalization of the well-known Hamiltonian walk problem (when $k=1$), $k$-HWP is NP-hard. 
We propose a $(3-\frac{1}{21})$-approximation algorithm for 2-HWP on arbitrary graphs.
For the tree, we define a restricted version of the problem and present an optimal algorithm for arbitrary values of $k$.
Finally, we formalize the problem for $k$-uniform hypergraphs and present a $2(1+\ln k)$-approximation algorithm. 
This result is also adapted to design an approximation algorithm for $k$-HWP on general graphs when $k = O(1)$.
\end{abstract}


\section{Introduction}
\label{sec:Introduction}

Graph traversal is a fundamental problem in computer science and robotics, where the goal is to visit every node of a graph at least once. 
A key objective of graph traversal is to minimize the number of steps needed to traverse the entire graph.
Typically, such tasks can be accomplished by a single agent.
However, relying on a single agent can be inefficient, prone to failure, and unsuitable for large or complex environments.
These limitations motivate the need for multi-agent systems, where the agents coordinate among themselves to traverse the whole graph while reducing the traversal time and handling failures effectively.
One major challenge in coordinating multiple agents is ensuring safe and reliable communication during the traversal.
In many practical systems, such as swarm robotics or drone networks, connectivity among agents is essential to maintain real-time coordination, prevent collisions, and quickly react to unexpected failures.
In addition, the connectivity constraint also plays a key role in fault-tolerant settings. For instance, if one agent fails, the connected agents can detect the failure and either compensate for the loss or design a strategy to recover the lost functionality.

\noindent \textit{Preliminaries, Related works and Problem definition:} Let $G=(V, E)$ be a simple undirected connected graph (without self-loops and multiple edges), where $V$ is the set of vertices and $E$ is the set of edges. An edge between two vertices $v_i$ and $v_j$ is denoted by $(v_i, v_j)$. Let $n=|V|$ be the number of vertices and $m=|E|$ be the number of edges of $G$. $N(v)$ is the open neighbourhood of the vertex $v \in V$, consisting of all its neighbours, while $N[v]$ is the closed neighbourhood of $v$, defined as $N[v] = N(v) \cup \{v\}$.
A sequence of distinct vertices $P=\{v_1, v_2, \cdots, v_{p}\}$ is called a \textit{path} of the graph $G$ if $(v_i, v_{i+1}) \in E ~\forall i \in \{1, 2, \cdots, p-1\}$.
Whenever both endpoints of a path $v_1$ and $v_p$ are explicitly mentioned, we denote the path as $P_{v_1v_p}$. Additionally, if $v_1 = v_p$, the sequence is called a \textit{cycle} of the graph $G$.
If $G$ admits a path of length $(n-1)$, then the path is called a \textit{Hamiltonian path}. Similarly, if a graph has a cycle of length $n$, then the cycle is called a \textit{Hamiltonian cycle}, and the graph is called \textit{Hamiltonian graph}. In other words, a graph is called $\textit{Hamiltonian}$ if there exists a cycle that passes through each vertex of $G$ exactly once. The problem of finding a Hamiltonian cycle in a graph, or checking whether a given graph is Hamiltonian, is known as the \textit{Hamiltonian cycle problem}, which is one of Karp's 21 NP-complete problems \cite{Karp1972}. Given that a class of graphs lacks a Hamiltonian path or cycle, there is growing interest in determining a minimum-length spanning walk of those graphs. A \textit{walk} is a path with some possible repeated vertices. If both endpoints of the walk are the same (or different), the walk is termed a closed (or open) walk. A walk is a \textit{spanning walk} of $G$ if it passes through each vertex of $G$ at least once. A closed spanning walk of minimum length is termed as \textit{Hamiltonian walk}. The length of a Hamiltonian walk of the graph $G$ is known as \textit{Hamiltonian number} and denoted as $h^{+}(G)$. If $h^+(G) = n$, the graph $G$ is Hamiltonian and vice versa. $n \leq h^+(G) \leq 2(n-1)$. It is well known that $h^+(G) = 2n-2$ iff $G$ is a tree \cite{goodman1974hamiltonian}.
An open spanning walk of $G$ is a spanning walk with different endpoints, and we use the notation $h^-(G)$ to denote the length of the minimum open spanning walk of $G$. 
Finding the minimum-length open spanning walk of a graph is known as the \textit{open Hamiltonian walk problem}.
If the graph $G$ is Hamiltonian, then $h^-(G) = (n-1)$. However, the converse is not true.

The challenge of finding an open spanning walk is closely related to the classic \textit{traveling salesman problem }(TSP), in which the objective is to locate a minimum weight Hamiltonian cycle in a given complete weighted graph.
Christofides-Serdyukov \cite{Christofides} \cite{serdyukov1978nekotorykh} proposed a $\frac{3}{2}$-approximation algorithm for TSP if the weight function satisfies the triangular inequality. After decades, Karlin et al. \cite{karlin2021slightly} give a randomized $\frac{3}{2} - \epsilon$ approximation algorithm for some $\epsilon > 10^{-36}$. A polynomial time approximation scheme (PTAS) exists for some special instances of the TSP, such as for Euclidean TSP \cite{arora1996polynomial}, \cite{mitchell1999guillotine}, planar TSP \cite{arkin1998localPacking}, \cite{grigni1995approximation}, \cite{klein2005linear}, and bounded genus TSP \cite{demaine2010approximation}. The \textit{multiple traveling salesman problem} \cite{cheikhrouhou2021comprehensive} is a generalization of the TSP, where multiple salesmen independently (no restriction to the connectivity) visit the vertices of the graph. Hoogeveen \cite{HOOGEVEEN1991291} modified Christofide's heuristic to find a minimum weight Hamiltonian path of a given weighted complete graph when the weights satisfy the triangular inequality. The modified heuristic also gives a $\frac{3}{2}$ approximation factor, which is still the best-known factor for the problem. We can apply the modified Christofides heuristic to find an open spanning walk for an undirected, unweighted, and connected graph $G$ by converting it into a complete weighted graph. To do this, we assign a weight equal to one to each existing edge and add all missing edges with a weight equal to the length of the shortest paths between corresponding vertices.

We address the open Hamiltonian walk problem from the perspective of a mobile agent that explores the entire graph $G$ in a minimum number of steps, starting from an arbitrary vertex. At each step, the agent can move to one of its neighbouring vertices. Therefore, any minimum open-spanning walk constitutes a feasible solution. Additionally, we can apply the modified Christofide's heuristic \cite{HOOGEVEEN1991291} to obtain a $\frac{3}{2}$-approximation algorithm for the problem. In this paper, we formulate the generalized version of the Hamiltonian walk problem, which we named as $k$-\textit{agent Hamiltonian Walk problem} (k-HWP), where a set of $k~(\geq 1)$ agents collectively explore the graph. The constraint is that at any instant of time $t(\geq0)$, $k$ agents must lie on $k$ distinct vertices of $G$, so that the subgraph induced by the vertices occupied by the $k$ agents must be connected. 
Initially, the agents can start from any arbitrary vertices that satisfy the above constraints.
In the next step, all agents move synchronously to one of their neighbouring vertices or choose to remain stationary. Let $A = \{a_1, a_2, \cdots, a_k\}$ be the set of $k$ agents. 
We use $v_i^t$ to denote the vertex $v_i$ occupied by a agent $a_i$ at the $t$-th step.
We denote the \textit{configuration} at the $t$-th step by $\mathcal{C}_t$, and defined by the $k$-tuple $(v_1^t, v_2^t, \cdots, v_k^t)$, where $v_i^t \neq v_j^t, \forall i \neq j$ and the subgraph induced by the set of vertices $\{v_1^t, v_2^t, \cdots, v_k^t\} \subseteq V$ is connected.
For convenience, we slightly abuse the notation $\mathcal{C}_t$ to denote the set $\{v_1^t, v_2^t, \cdots, v_k^t\}$ and, more generally, any arbitrary configuration. Two configurations, $\mathcal{C}_t$ and $\mathcal{C}_{t'}$, are said to be \textit{adjacent} if $\forall i \in \{1, 2, \cdots, k\}$ either $v_i^t = v_i^{t'}$ or $(v_i^t, v_i^{t'}) \in E$.
We imagine a graph whose vertices correspond to configurations, where two vertices are adjacent if the corresponding configurations are adjacent. We refer to the edge between two adjacent vertices, i.e., configurations, as a \textit{transition edge}.
The objective of $k$-HWP is to find a sequence of configurations $\{\mathcal{C}_0, \mathcal{C}_1, \cdots, \mathcal{C}_l\}$ with a \textit{minimum length} (i.e., $l$ is minimum) such that each vertex of $G$ appears in at least one configuration in the sequence, and there is a transition edge between every two consecutive configurations. In other words, the objective is to find an open spanning walk of configurations $\{\mathcal{C}_t\}^l_{t=0}$ of minimum length, where each vertex in this walk is a $k$-tuple configuration and $\cup_{t=0}^l \cup_{i=1}^k v_i^t = V$. We denote the length of this open walk by $h^-_k(G)$. If there are $(l+1)$ many configurations in the walk, then $h^-_k(G) = l$. When $k=1$, $h_1^-(G)$ coincides with $h^-(G)$.
An open walk of configurations is called a \textit{ transition walk}.
Since we focus solely on the open walk in this paper, we use $h_k(G)$, instead of $h^-_k(G)$, to denote the length of the minimum open spanning walk of configurations for the graph $G$ with $k$ agents. The NP-completeness of $k$-HWP follows from the fact that when $k=1$, it is equivalent to the Hamiltonian walk problem, which is known to be NP-complete.
We present the formal definition of $k$-HWP below.
\begin{definition}[Problem: $k$-HWP]
 Consider an undirected connected graph $G$ and a collection of $k$ agents. The goal is to find a sequence of configurations of minimum length such that each successive pair in the sequence is adjacent, and every vertex in $G$ is included at least once in one of these configurations.
\end{definition}
When $k=2$, any configuration represents an edge of $G$. Therefore, we can interpret the problem as traversing the graph through its edges.
\begin{definition}[$r$-transition edge]
  The transition edge between two configurations $\mathcal{C}_t$ and $\mathcal{C}_{t'}$ is said to be a $r$-transition edge if $\mathcal{C}_{t'}$ has exactly $r$ vertices that are not in $\mathcal{C}_t$. 
\end{definition}
For example of a $r$-transition edge, let $k=4$, $\mathcal{C}_t= (v_1, v_2, v_3, v_4)$ and $\mathcal{C}_{t'} = (v_5, v_2, v_4, v_6)$. If $\mathcal{C}_t$ and $\mathcal{C}_{t'}$ are adjacent, then the edge between them is called a $2$-transition edge.
The transition edges of $k$-HWP can be categorized into $k$ distinct classes based on the value of $r \in \{1,2, \cdots, k\}$. 
We disregard $0$-transition edges, which indicate that agents either remain stationary or rearrange their positions without visiting new vertices. Since our objective is to identify a transition walk of minimum length, we consistently strive to minimize the number of configurations encountered during any transition walk.
For example, let $\mathcal{C}_t= (v_1, v_2, v_3, v_4)$, $\mathcal{C}_{t'} = (v_2, v_4, v_3, v_1)$ and $\mathcal{C}_{t''} = (v_5, v_6, v_7, v_8)$. Suppose that a transition walk includes the configuration $\mathcal{C}_t$, followed by $\mathcal{C}_{t'}$, and then $\mathcal{C}_{t''}$. The transition edge between $\mathcal{C}_t$ and $\mathcal{C}_{t'}$ is $0$-transition edge. We can modify the walk by removing $\mathcal{C}_{t'}$ and replacing $\mathcal{C}_{t''}$ with $\mathcal{C}_{t'''}=(v_8, v_5, v_7, v_6)$. The modification is valid because the configurations $\mathcal{C}_{t'}$, $\mathcal{C}_{t''}$ are adjacent, implies $(v_1, v_8), (v_2, v_5), (v_3, v_7), (v_4, v_6)$ all belonging to $E$, implies $\mathcal{C}_t$, $\mathcal{C}_{t'''}$ are adjacent. By making this change, we eliminate the $0$-transition edge.

\subsection{Our Contribution}
In this paper, our contributions are the following:
\begin{itemize}
    \item We define $k$-agents Hamiltonian walk problem ($k$-HWP).
    \item For the tree, we propose an optimal algorithm when $k \leq 3$. For any value of $k$, we present an optimal algorithm for the restricted version of the problem on tree, where we only allow 1-transition edges (Theorem \ref{thm:OptimalBound_k-RHWP_Tree}).
    \item We give a $(3-\frac{1}{21})$-approximation algorithm on arbitrary graphs for $k=2$ (Theorem \ref{thm:<3-2-HWP}).
    \item We extend the problem definition to hypergraphs and propose a $2(1+\ln k)$-approximation algorithm for $k$-uniform hypergraphs (Theorem \ref{thm:approx-hypergraph}). 
    This algorithm is further developed into an approximation algorithm for $k$-HWP, when $k=O(1)$ (Theorem \ref{thm:k-HWP-graph}).
\end{itemize}

\section{Algorithms for Tree}
\label{sec:tree}
Since the $k$-HWP problem is NP-hard for general graphs, our goal is to explore specific graph classes where optimal solutions can be achieved. 
This section examines $k$-HWP problem on acyclic graphs and presents results for a constrained version of the problem.
First, we describe an optimal strategy (\textsc{1-HWP-For-Tree}) to solve 1-HWP for tree. This serves as a foundation for understanding the generalized strategy (for arbitrary values of $k$), which we discuss later in the section. A vertex is called \textit{leaf} (or \textit{pendant}) if it has degree one. The diameter of the graph $G$, denoted $diam(G)$, is the length of the longest path.

\noindent \textbf{Description of the Algorithm} (\textsc{1-HWP-For-Tree}): Let $P^* = P_{v_1^* v_2^*}$ be a longest path of the tree $G$. Then, both $v_1^*$ and $v_2^*$ must be leaves of $G$. The agent starts the traversal from $v_1^*$ and aims to eventually reach $v_2^*$ while traversing all the vertices. If multiple unexplored vertices are available, it prioritizes those not located on $P^*$. Whenever all adjacent vertices are explored, it backtracks to its parent vertex. A parent vertex of $v$ is the vertex from which the vertex $v$ is explored for the first time. The pseudocode of the Algorithm is provided in Algorithm \ref{alg:1-HWP-Tree}. We analyze the algorithm below.

\begin{algorithm}
	\caption{\textsc{1-HWP-For-Tree}}	
	\label{alg:1-HWP-Tree}
	\begin{algorithmic}[1]
        \State {Find the longest path of the graph $G$, say $P^* = P_{v_1^* v_2^*}$.}
        \For{every $v$ in $V$}
        \State {$explored(v) = FALSE$ and $parent(v) = NULL$}
		\EndFor
        \State {Start from the vertex $v_1^*$ and set $v = v_1^*$; $explored(v) = TRUE$}
        \While{$v \neq v_2^*$}
        \If{$\exists$ a $v' \in N(v)$ with $explored(v') = FALSE$ and $v' \notin P^*$}
        \State{Move to the vertex $v'$}
        \State{$explored(v') = TRUE;~ parent(v') = v;~ v = v'$}
        \ElsIf{$\exists$ a $v' \in N(v)$ with $explored(v') = FALSE$ and $v' \in P^*$}
        \State{Move to the vertex $v'$}
         \State{$explored(v') = TRUE;~ parent(v') = v;~ v = v'$}
        \Else
        \State{Move to $parent(v)$ and set $v = parent(v)$}
        \EndIf
        \EndWhile
	\end{algorithmic}
\end{algorithm}

\begin{lemma}
    \label{lem:upperBound_1HWP_Tree}
    The walk returned by \textsc{1-HWP-For-Tree} is a spanning walk of $G$ with length $2(n-1) - diam(G)$, where $diam(G)$ is the diameter and $n$ is the number of vertices of $G$.
\end{lemma}
    
\begin{proof}
    Steps (6 - 14) from the pseudocode of Algorithm \ref{alg:1-HWP-Tree} justify that the agent lying on an explored vertex $v$ backtracks (move to the parent) only when it finds all its neighbouring vertices explored. However, an agent lying on $P^*$ never backtracks, as it always finds an unexlored neighbour lying on $P^*$. So, the existence of unexplored vertex $(\notin P^*)$ means that the agent backtracks despite having an unexplored neighbour, which is a contradiction. Additionally, all the vertices on $P^*$ must be explored as there is a unique way to reach $v_2^*$ from $v_1^*$, which is through $P^*$.
    
    Let $(x, y)$ be an arbitrary edge in $E$. Without loss of generality, we assume that $x$ is the parent of $y$. If $(x, y) \in P^*$, the agent never moves back to $x$ from $y$, as it must encounter an unexplored neighbour vertex $z \in P^* \cap N(y)$ and moves to it (steps 10-11). If $(x, y) \notin P^*$, then $y$ also does not lie in $P^*$. Here, when the agent reaches $y$, it returns to $x (parent(y))$ after confirming that every vertex in $N(y)$ has been explored (steps 13-14). Consequently, all the edges that do not belong to $P^*$ are traversed twice in the walk produced by the algorithm. Therefore, the length of the walk is $2(n-1) - diam(G)$, as the tree has $(n-1)$ edges and the number of edges in the longest path equals the diameter of the tree.
\qed \end{proof}
\begin{remark}
\label{rem:time-1-HWP-tree}
    Algorithm \ref{alg:1-HWP-Tree} runs polynomially. This is because we can find the diameter of the graph in linear time ( in terms of $n$). All other operations per transition also take polynomial time, since all checks and assign operations for a vertex are done in at most $|N(v)| \leq \Delta$ time, where $\Delta$ is the maximum degree of the tree $G$. So the algorithm needs in total $O(n) + \sum O(\Delta)$ time, where the summation is over number of steps required by the algorithm, which is at most $2n$. Hence the overall time complexity of the algorithm is $O(n \cdot \Delta)$ time.
\end{remark}

\begin{lemma}
    \label{lem:OptimalBound_1HWP_Tree}
    For a tree $G$ with $n$ vertices, $h_1(G) \geq 2(n-1) - diam(G)$.
\end{lemma}

\begin{proof}
    Any open spanning walk must contain each edge of the tree at least once.
    Consider an optimal spanning walk $W$, which starts from the vertex $v_i$ and ends at $v_j$. Since the graph is a tree, there is always a unique path $P$ between $v_i$ and $v_j$. Let $(x, y) \in E$ be an edge that does not belong to $P$. We will prove that $(x, y)$ must appear twice in this walk. Then $W$ is the union of either (i) the sub-walk $v_i$ to $x$ with no internal $y$, then $x$ to $y$, and finally the sub-walk $y$ to $v_j$, or (ii) the sub-walk $v_i$ to $y$ with no internal $x$, then $y$ to $x$, and then the sub-walk $x$ to $v_j$. For the first case, the edge $(x, y)$ is encountered for the first time when it traverses from $x$ to $y$. Then again when traverses from $y$ to $v_j$, as there is an unique path between $y$ and $v_j$ which contains $x$. Similarly, we can prove for the second case. However, any edge lying on $P$ may not meet either of the above two conditions. So, the length of the walk $W$ is $len(W) = h_1(G) \geq $ (total number of edges) + (number of edges not in $P$) $\geq (n-1)+(n-1-diam(G))$, since the length of $P$ is at most $diam(G)$.
\qed \end{proof}

Hence we have an algorithm (\textsc{1-HWP-For-Tree}) for 1-HWP on tree, which returned a spanning walk that matches with the lower bound of $h_1(G)$. Therefore, combining Lemma \ref{lem:upperBound_1HWP_Tree},  \ref{lem:OptimalBound_1HWP_Tree} and Remark \ref{rem:time-1-HWP-tree}, we can state the following theorem.

\begin{theorem}
\label{thm:OptimalBound_1HWP_Tree}
    Algorithm \textsc{1-HWP-For-Tree} provides an optimal solution for 1-HWP in polynomial time, when the input graph $G$ is a tree.
\end{theorem}

\noindent Using the above strategy, we now design an algorithm for a restricted version of $k$-HWP on trees, where the restriction is on the type of transition edges allowed between two configurations. We begin with the following structural property.

\begin{theorem}
    \label{thm:1-transit-edge_tree}
    If $G$ is a tree and the edge between two configurations $\mathcal{C}_t$ and $\mathcal{C}_{t'}$ in the $k$-HWP is an $r$-transition edge, then $r \leq \lfloor \frac{k}{2} \rfloor$, where $k \geq 2$.
\end{theorem}

\begin{proof}
    We first prove that there must exist a vertex $v$ in $\mathcal{C}_t \cap \mathcal{C}_{t'}$.
    Suppose, on the contrary, that $\mathcal{C}_t \cap \mathcal{C}_{t'} = \emptyset$.
    In such a scenario, there exist two distinct paths between the vertices $v_i^t$ and $v_j^t$, where $v_i^t$ and $v_j^t$ are the positions of the agents $a_i$ and $a_j$, respectively, in the configuration $\mathcal{C}_t$. 
    Since the induced graph $G[\mathcal{C}_t]$ made by the vertices of $\mathcal{C}_t$ is connected, there exists a path between $v_i^t$ and $v_j^t$ that entirely lies in $G[\mathcal{C}_t]$.
    An alternative path proceeds via the edge $(v_i^t, v_i^{t'})$, followed by a path within $G[\mathcal{C}_{t'}]$ from $v_i^{t'}$ to $v_j^{t'}$, and concludes with the edge $(v_j^{t'}, v_j^t)$. The existence of two such distinct paths contradicts the fact that there is always a unique path between two vertices in a tree.

    Now consider a vertex $v_i^{t'} \in \mathcal{C}_{t'} \setminus \mathcal{C}_{t}$. We claim that $v_i^t \in \mathcal{C}_{t'}$.
    Otherwise, there would exist two paths between some $v \in \mathcal{C}_t \cap \mathcal{C}_{t'}$ and $v_i^{t'}$: one passing through $v_i^t$ and the other does not, which is again a contradiction.
    Hence, for every $v_i^{t'} \in \mathcal{C}_{t'} \setminus \mathcal{C}_{t}$, the corresponding $v_i^t \in \mathcal{C}_t \cap \mathcal{C}_{t'}$. In other words, for each new vertex of the current configuration ($\mathcal{C}_{t'}$), we always have a unique vertex that is part of both the current and previous configuration ($\mathcal{C}_t$). Thus the proof.
\qed \end{proof}

The above result implies that when $k \leq 3$, all transition edges are necessarily $1$-transition edges. 
For larger values of $k$, between the movements of one configuration to the other, the agents might encounter $k/2$ new vertices.
To simplify the transition structure and enable the design of an optimal algorithm, we consider a restricted version of $k$-HWP, where each transition edge must be a $1$-transition edge, regardless of the number of agents $(k)$.
We refer to this problem as $k$-\textit{agents restricted Hamiltonian walk problem} ($k$-RHWP). 
We present a polynomial-time algorithm, \textsc{$k$-RHWP-For-Tree}, that computes an optimal solution to the $k$-RHWP on trees.
The optimal length of the $k$-RHWP on $G$ is denoted by $h^r_k(G)$.
Note that $h^r_2(G) = h_2(G)$ and $h^r_3(G) = h_3(G)$.
We introduce subsequent definitions (Definition  \ref{def:k-RHWP}, \ref{def:encountered}
), which are limited to the tree.
\begin{definition}[Problem: $k$-RHWP]
\label{def:k-RHWP}
    Consider an undirected connected graph $G$ along with a collection of $k$ agents. The objective is to find a sequence of configurations of minimum length such that between every two consecutive configurations in the sequence there is a 1-transition edge, and every vertex of $G$ is included at least once in one of these configurations. 
\end{definition}
\begin{definition} [Explored edge or vertex]
\label{def:encountered}
Let $\mathcal{C}_t$ and $\mathcal{C}_{t'}$ be two adjacent configurations in a transition walk of $k$-RHWP; then there exists exactly one vertex $v_i^{t'} \in \mathcal{C}_{t'} \setminus \mathcal{C}_t$. We say that the edge $(v_i^t, v_i^{t'})$ is explored (or encountered) by the transition from $\mathcal{C}_t$ to $\mathcal{C}_{t'}$ or the vertex $v_i^{t'}$ is explored (or encountered).
\end{definition}

\noindent Next, we describe the algorithm $k$\textsc{-RHWP-For-Tree} for solving $k$-RHWP on trees.
We designate one of the agents as the head, who explores the graph. All other agents either follow the head or remain stationary to maintain connectivity.

\noindent \textbf{Description of the Algorithm} ($k$\textsc{-RHWP-For-Tree}): 
Let $P^* = P_{v_1^* v_2^*}$ be a longest path of $G$.
All vertices are initially unexplored.
Agents are initially deployed using the strategy \textsc{1-HWP-For-Tree}. The agent $a_i$ is placed at the vertex $v_i^0$, the $i$-th explored vertex by \textsc{1-HWP-For-Tree}. So, the initial configuration $\mathcal{C}_0$ equals $(v_1^0, v_2^0, \cdots, v_k^0)$, where $v_1^0 = v_1^*$.
We refer to the vertices $v_k^0$ and $v_1^0$ as the \textit{head} $(v_H)$ and the \textit{tail} $(v_T)$ of the current configuration, respectively.
Additionally, we apply the terms head and tail to agents at $v_H$ and $v_T$, respectively.
The head takes priority in exploring the graph.
If the head encounters a situation where it finds all its neighbours as explored, it locates the nearest occupied vertex that has at least one unexplored neighbour.
This vertex then becomes the new head, while the previous head is designated as the new tail.
If the head has multiple unexplored neighbours, it selects the vertex $v$, which does not lie on $P^*$, as its next target of movement and sets $parent(v) = v_H$.
If the neighbours of all the occupied vertices (including head) in the current configuration are explored, the agent $a_i$ closest to $P^*$ becomes the new head, while the previous head becomes the new tail.
In this case, agent $a_i$ decides its target as $parent(v')$, where $v'$ is the current position of $a_i$.
Once the head sets its target for the next movement, all the agents lying on the path between head and tail, including the tail, will move simultaneously in the direction of the head, while all other agents remain stationary.
If the tail is not a leaf vertex of the induced graph $G'$, which is formed by all occupied vertices, then moving the tail would disconnect $G'$.
In such a case, we assign the tail to one of the leaf vertices of $G'$, excluding both the head and the vertex nearest to $P^*$, if possible.
Once all the agents between $v_H$ and $v_T$ have moved to their designated targets, we update the head and tail accordingly. Whenever the vertex $v_2^*$ becomes the head, the algorithm terminates. The pseudocode is mentioned in Algorithm \ref{alg:k-HWP-Tree}.

\noindent \textbf{Execution Example of the Algorithm \textsc{k-RHWP-For-Tree}} Consider the graph as depicted in Fig. \ref{fig:exec-ex}. The vertex set is $V= \{v_1, v_2, \cdots, v_{18}\}$. Let $P^* = P_{v_1v_{18}}$. The solution returned by \textsc{1-HWP-For-Tree} is the walk $W_1 = \{v_1, v_2, v_3, $ $ v_4, v_5, v_6, v_7, v_6, v_8, v_9, v_8, v_6, v_5, v_{10}, v_{11}, v_{12},$ $ v_{13}, v_{12}, v_{14}, v_{12}, v_{11}, v_{10}, v_5, v_{15}, v_{16},$ $ v_{17}, v_{18}\}$. Now, we investigate the problem $4$-RHWP and the solution returned by the algorithm $k$\textsc{-RHWP-For-Tree}. The initial configuration is $\mathcal{C}_0 = (v_1, v_2, v_3,$  $ v_4)$, since these are the four vertices explored by $W_1$. $v_4$ is termed head and $v_1$ is termed tail. In the next step, the head $(v_4)$ moves to $v_5$ and all the other agents follow the head. After the movement, $\mathcal{C}_1 = (v_2, v_3, v_4, v_5)$. Similarly, we have $\mathcal{C}_2 = (v_6, v_5, v_4, v_3)$ and $\mathcal{C}_3 = (v_7, v_6, v_5, v_4)$. At $\mathcal{C}_3$, the current head $(v_7)$ does not have an unexplored neighbour. The head is transferred to $v_6$, which is the nearest occupied vertex with at least one unexplored neighbour. In addition, the previous head $v_7$ becomes the new tail. Now, the head sets the target to $v_8$. All agents between head and tail, which is in this case tail only, follow the head, that is, move to $v_6$. However, the other two agents remain stationary. After the movement, the configuration is $\mathcal{C}_4 = (v_6, v_8, v_5, v_4)$, where $v_8$ is the head and $v_6$ is the tail. Currently, the tail $v_6$ is the non-leaf vertex of the induced graph $G'$ formed by the set of occupied vertex $\{v_6, v_8, v_5, v_4\}$. In this scenario, the tail is transferred to $v_4$, which is the only leaf vertex (excluding the head $v_8$). After this, all agents set their target and reach the configuration $\mathcal{C}_5=\{v_8, v_9, v_6, v_5\}$. Similarly, we can have $\mathcal{C}_6 =\{v_6, v_8, v_5, v_{10}\}; \mathcal{C}_7 =\{v_5, v_6, v_{10}, v_{11}\}; \mathcal{C}_8 =\{v_{10}, v_5, v_{11}, v_{12}\}; \mathcal{C}_9 =\{v_{11}, v_{10}, v_{12}, v_{13}\}; \mathcal{C}_{10} =\{v_{11}, v_{10}, v_{14}, v_{12}\}$. In this situation, the head is at $v_{14}$, the tail is at $v_{12}$, and there is no occupied vertex with an unexplored neighbour. $v_10$, the nearest vertex to $P^*$ becomes the new head and sets its target as $parent(v_{10} = v_5$. We can similarly compute the other configurations, which are  
$\mathcal{C}_{11} =\{v_{10}, v_5, v_{12}, v_{11}\}; \mathcal{C}_{12} =\{v_5, v_{15}, v_{11}, v_{10}\}; \mathcal{C}_{13} =\{v_{15}, v_{16}, v_{10}, v_{5}\}; \mathcal{C}_{14} =\{v_{16}, v_{17}, v_{5}, v_{15}\}; \mathcal{C}_{15} =\{v_{17}, v_{18}, v_{15}, v_{16}\}$. In the last transition, the vertex $v_{18}$ is encountered, and so the algorithm terminates. The length of the transition walk is 15.
\begin{figure}[h]
    \centering   \includegraphics[width=1\linewidth]{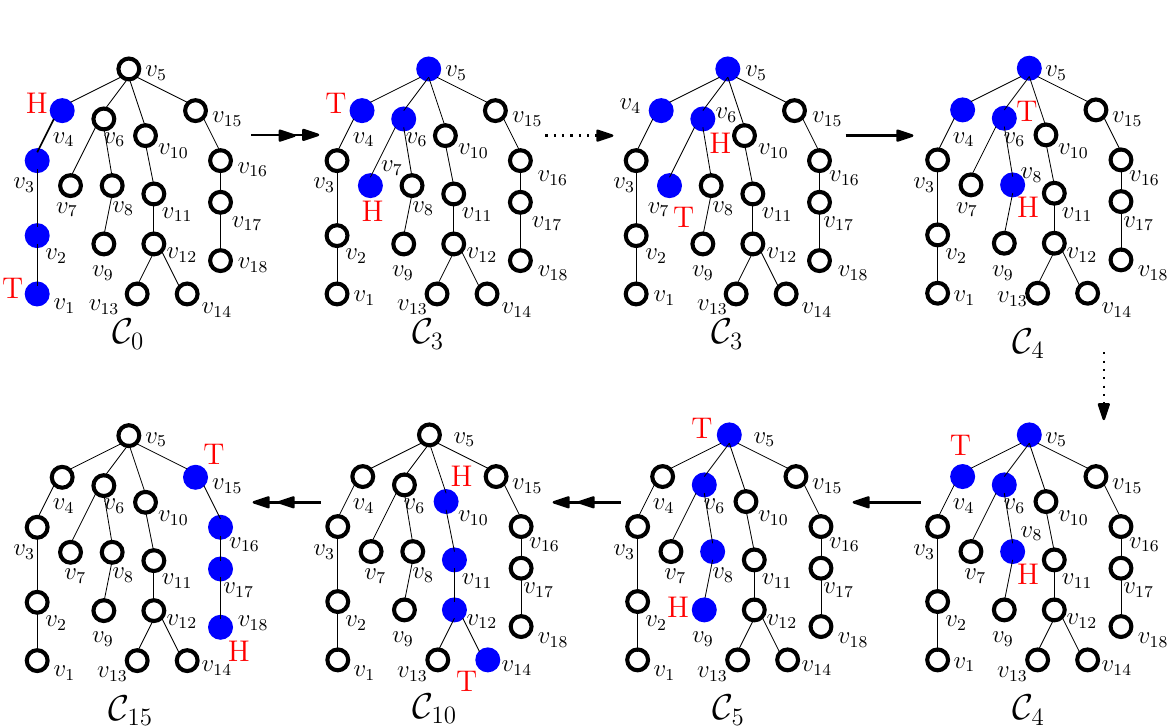}
    \caption{Illustrating a stepwise execution of $k$\textsc{-RHWP-For-Tree} with $P^* = P_{v_1v_{18}}$ and $k=4$. The \textcolor{blue}{blue} nodes represent the occupied vertices while \textcolor{red}{H} and \textcolor{red}{T} indicate the head and tail, respectively. $\mathcal{C}_i$ is the configuration at $i$-th step. Solid single arrows represent direct transitions between configurations; solid double arrows indicate intermediate configurations are omitted; dotted arrows represent the transposition of head and tail within a configuration}
    \label{fig:exec-ex}
\end{figure}

\noindent \textbf{Analysis of the Algorithm \textsc{k-RHWP-For-Tree:}}
We analyze the algorithm \textsc{k-RHWP-For-Tree} in the following after introducing some definitions, which are limited for the case of tree.
\begin{definition}
\label{def:d_P}
For a fixed path $P$ and set of $k$ agents, we define $d_P^k(e)$ for each edge $e = (u, v) \in E$ as follows: $d_P^k(e) = 0$,  $\forall e \in P$. For all the other edges $e \notin P$, let $u$ be the vertex closer to $P$ than $v$ and let $T_P(e)$ be the subtree rooted at $v$ after deleting $e$ from $G$. If the height of the tree $T_P(e)$ is at most $(k-2)$, we set $d_P^k(e) = 0$ otherwise $d_P^k(e) = 1$.
In other words, $d_P^k(e) = 1$, if $\exists z \in T_P(e)$ such that the distance between $u$ and $z$ is at least $k$.
\end{definition}

 For example, in Fig. \ref{fig:exec-ex}, consider the path $P=P_{v_1v_{18}}$ and $k = 4$. $d_P^4(e' =(v_5, v_{10})) = 1$, as the height of $T_P(e)$ rooted at $v_{10}$ is $3 > (k-2)$. For all other edges $e$, $d_P^4(e) = 0$.  $d_P^k(e) = 0$ indicates that the edge $e$ is encountered exactly once, while $d_P^k(e) = 1$ means that it is encountered more than once. The edge $(v_5, v_{10})$ is encountered twice by our algorithm; first during the transition from $\mathcal{C}_5$ to $\mathcal{C}_6$ and again during the transition from $\mathcal{C}_{10}$ to $\mathcal{C}_{11}$.

\begin{algorithm}
	\caption{k-\textsc{RHWP-For-Tree}}
	\label{alg:k-HWP-Tree}
	\begin{algorithmic}[1]
        \Statex{/* \textit{Initial deployment of the agents} */}
        \State{Find the longest path of the tree, say the path is $P^*$ with the endpoints $v_1^*$ and $v_2^*$}
        \For{every $v \in V$}
        \State{$explored(v) \gets FALSE$}
        \EndFor
        \State{Deploy the agent $a_1$ at $v_1^*$}
        \State{Use \textsc{1-HWP-For-Tree} to explore exactly $k$ vertices starting from $v_1^*$, let the vertex $v_i^0$ is the $i$-th explored vertex.}
        \State{Deploy the agent $a_i$ at the vertex $v_i^0$, $2 \leq i \leq k$}
        \State{Set $v_H \gets v_k^0$; $v_T \gets v_i^0 = v_1^* $}
        \Statex{/* Transition Part */}
        \While{$v_H \neq v_2^*$}
        \State{Let $V' \subseteq V$ be the set of currently occupied vertices by the $k$ agents, and let $G'$ be the corresponding induced subgraph of $G$}
        \Statex{/* Check whether there exists a vertex $v \in V'$ with one unexplored neighbour */}
        \If{($\exists~ v' \in N(v_H)$ with $explored(v') = FALSE$) or (($\forall~ v' \in N(v_H),~ explored(v') $ $= TRUE$) and ($\exists$ a vertex in $G'$ other than $v_H$ with one unexplored neighbour))}
        \If{$\nexists v' \in N(v_H)$ with $explored(v') = FALSE$}
        \State{Find the vertex $v$ nearest to $v_H$ in $G'$ with at least one unexplored neighbour}
        \State{$v_T \gets v_H;~ v_H \gets v$}
        \EndIf
        \If{$\exists~ v' \in N(v_H)$ with $explored(v) = FALSE$ and $v' \notin P^*$}
        \State{The agent, that lie on $v_H$, set its target $v_H^{tgt} = v'$}
        \Else \Comment{$\exists~ v' \in N(v_H)$ with $explored(v) = FALSE$ and $v' \in P^*$}
        \State{The agent, that lie on $v_H$, set its target $v_H^{tgt} = v'$}
        \EndIf
        \Else \Comment{all the neighbours of $v$ are explored, $\forall v \in V'$}
        \State{Find the vertex $v \in V'$ which is nearest to $P^*$}
        \If{$v \neq v_H$}
        \State{$ v_T \gets v_H; ~ v_H \gets v$}
        \EndIf
        \State{The agent, that lie on $v_H$, set its target $v_H^{tgt} = parent(v_H)$}
        \EndIf
        \State{The agent, that lie on $v_H$, moves to the target vertex $v^{tgt}_H$}
        \If{$v_T$ is not a leaf vertex of $G'$}
        \Comment{change the tail to a leaf in $G'$}
        \State{Find the set of leaves $L'$ in $G'$ and the vertex $u' \in L'$ that is nearest to $P^*$}
        \If{there exists a vertex $u'' \in L'$ other than $u'$}
        \State{$v_T \gets u''$}
        \Else
        \State{$v_T \gets u'$}
        \EndIf
        \EndIf
        \State{Let $P = \{v_H = v_1, v_2, \cdots, v_p=v_T\}$ be the path between $v_H$ and $v_T$}
        \Statex{/* All the agents on the path $P$ follow the head */}
        \State{The agent lie on $v_i$ sets its target $v_i^{tgt}=v_{i-1} (\forall i: 2 \leq i \leq p)$ and all the other agents not lying on $P$ remain stationary}
        \State{All the agents move to their respective targets in one step}
        \State{$explored(v_H^{tgt}) = TRUE$, $v_H \gets v_H^{tgt}$, $v_T \gets v_{p-1}$}
        \EndWhile
	\end{algorithmic}
\end{algorithm}
\begin{definition}
\label{def:P-Sol}
Let $Sol$ be an algorithm for $k$-RHWP. We define the path $P_{Sol}$ as the path between the vertices $x_{Sol}$ and $y_{Sol}$, where $y_{Sol}$ is the last vertex encountered or explored by the  $Sol$, and $x_{Sol}$ is the vertex in the initial configuration of $Sol$ that is farthest from $y_{Sol}$. 
\end{definition}

\begin{lemma}
    \label{lem:Correctness-k-HWP}
    The walk of configurations returned by the algorithm $k$\textsc{-RHWP-For-Tree} is a spanning walk of $G$ of length $t_{ALG}$ = $(n- k) + \sum_{e \in E} d^k_{P^*}(e)$, where $P^*$ is the longest path of $G$ between the vertices $v_1^*$ and $v_2^*$ selected by the algorithm.
\end{lemma}

\begin{proof}
    The initial deployment of the agents encountered exactly $k$ vertices, which is assured by the correctness of Algorithm \textsc{1-HWP-For-Tree}. Similar to Algorithm \textsc{1-HWP-For-Tree}, we only move backward to the parent vertex when all the neighbours of the $k$-agents are explored, which confirms that each vertex must be in either the initial configuration or encountered by some subsequent configuration.
    To prove the second part of the lemma, it is sufficient to prove that any edge with $d^k_{P^*}(e) = 0$ encountered exactly once and an edge with $d^k_{P^*}(e) = 1$ encountered exactly twice. First, we will prove that any edge $e = (u, v)$ lies on $P^*$ encountered exactly once, where $v$ is nearer to $v_2^*$ than $u$. Let $v$ be explored for the first time by the configuration $\mathcal{C}_t$. Since an agent (head) lies at $u$ moves to $v \in P^*$ only when all other neighbours of $u$ are already explored. As the final target is $v_2^*$ and $v$ is nearer to $v_2^*$ than $u$, no agent being a head ever moves back to $u$ from $v$. Therefore, the edge $e$ is encountered exactly once. Now, we will prove that any edge $e=(u, v) \notin P^*$ can be encountered at most twice by the algorithm. The edge $e$ is encountered when we reach a configuration with the head as $u$, and then it sets its next target for the movement as $v$. After this, if the algorithm reaches a configuration with the head as $v$ and all the vertex in the configuration are explored, then at the next configuration, it encountered $e$ again by moving the agent lie at $v$ to $u$, which is the $parent(u)$. Afterward, we can not have a configuration with an agent lying on $u$ as the head and $v$ as its next target, as $v$ has already been explored. So, the algorithm counters an edge at most twice. Next we show that, an edge $e$, with $d^k_{P^*}(e) = 0$, encountered exactly once. We will discuss the extreme case. A similar argument will follow for the other case. Let us consider an edge $e = (x, y)$ such that the height of the sub-tree $T_P(e)$ rooted at $y$ is exactly $(k-2)$. Let $e$ be first encountered by the configuration $\mathcal{C}_t$. All the later configurations that encountered the other vertices of $T_P(e)$ must contain the edge $e$, as we never make the vertex $x$ (which is nearest to $P^*$ among all the occupied vertices) as the tail ($v_T$) of those configurations (Step 33-34). Let's consider an example as referred in Fig. \ref{fig:ex-1-Tree}. Let $k=4$, $\mathcal{C}_t=(z, x, y_4, y_6)$ be the $t$-th configuration returned by the Algorithm $k$-\textsc{HWP-For-Tree}, where $v_H = y_6, v_T = z$. Then, if we apply the algorithm we encountered $e=(x, y)$ at the next configuration $\mathcal{C}_{t+1} = (z, y, x, y_4)$ with $v_H = y, v_T = y_4$. The subsequent configurations are $\mathcal{C}_{t+2} = (z, y_1, y, x)$ with $v_H = y_1, v_T = x$; $\mathcal{C}_{t+3} = (z, y, y_2, x)$ with $v_H = y_2, v_T = y$; $\mathcal{C}_{t+4} = (x, y_2, y_5, y)$ with $v_H = y_5, v_T = x$; $\mathcal{C}_{t+5} = (x, y_3, y_2, y)$ with $v_H = y_2, v_T = x$. This shows that all the configurations contain the edge $e(x, y)$. In the next configuration $v_H$ changes to $x$ and $v_T$ changes to $y_3$ and moves to the next configuration  $\mathcal{C}_{t+6} = (z, y_2, y, x)$ with $v_H = z, v_T = y_2$. Thus the lemma follows, i.e., the length of the walk is $(n-1) - (k-1) - \sum_{e \in P^*}d^k_{P^*}(e)$, $(n-1)$ referred as the number edges to be encountered, $(k-1)$ is subtracted as those many are covered in the initial configuration and the last part say that those edges are encountered twice.
\qed \end{proof}

\begin{figure}[h]
\centering
\begin{minipage}[b]{0.52\linewidth}
\centering
    \includegraphics[width=0.7\linewidth]{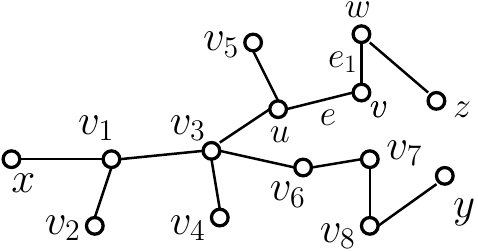}
    \caption{$P_{xy}$ is one of the longest paths.}
    \label{fig:define-d-Tree}
\end{minipage}
\hfill
\begin{minipage}[b]{0.46\linewidth}
    \centering
        \includegraphics[width=0.5\linewidth]{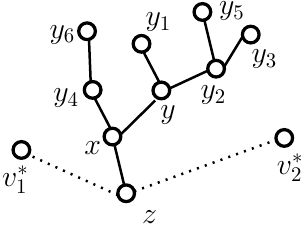}
        \caption{$v_1^*$ and $v_2^*$ are the endpoints of one of the longest paths of the tree.}
        \label{fig:ex-1-Tree}
\end{minipage}
\end{figure}

\begin{lemma}
    \label{lem:time-k-RHWP-for-tree}
    The time needed for the Algorithm \textsc{k-RHWP-For-Tree} is $O(n \cdot k \cdot \Delta)$, where $k$ is the number of agents, $n $ is the number of vertices, and $\Delta$ is the maximum degree of the tree $G$.
\end{lemma}

\begin{proof}
    The diameter of a tree can be computed in $O(n)$ time using the depth-first search (DFS) twice, which is a very standard process.
    First, perform a DFS from an arbitrary node (considered the root) to identify the node that is furthest from it. Then, perform a second DFS starting from that furthest node. The maximum distance discovered in this second traversal is the diameter of the tree.
    From Remark \ref{rem:time-1-HWP-tree}, the initial deployment can also be completed in $O(n)$ time.
    Thereafter, at each step, we need to identify an occupied vertex that has at least one unoccupied neighbour.
    For each such vertex $v$, we can check whether it has any unexplored neighbours—or determine that all its neighbours are already explored—in $O(\Delta)$ time, where $\Delta$ is the maximum degree of the graph.
    In the worst case, none of the occupied vertices have unexplored neighbours.
    In that case, we have to check for all the occupied vertices starting from the head, which requires $O(k\Delta)$ time, as there are $k$ many occupied vertices in each step (configuration).
    In addition, if the tail is not located at the the leaf of $G'$, the induced subgraph made by the occupied vertices, then we have to transfer the tail. 
    The construction of $G'$ and the execution of the transfer (if needed) again take $O(k)$ time.
    Therefore, each transition from one configuration to the next can be completed in at most $O(k\Delta)$ time.
    From Lemma \ref{lem:Correctness-k-HWP}, the algorithm terminates after at most $(n- k) + \sum_{e \in E} d^k_{P^*}(e) \leq (n-k)+ (n-diam(G)) < 2n$ steps. The total time complexity of the algorithm is $O(n \cdot k \cdot \Delta)$.
\qed \end{proof}

\begin{lemma}
    \label{lem:lower-Bound_k-Tree}
    Let $t_{Opt}$ be the length of an open spanning walk of configurations with $k$ agents, returned by an optimal algorithm $Opt$. Then $t_{Opt} \geq (n-k) + \sum_{e \in E} d^k_{P_{Opt}}(e).$ 
\end{lemma}
\begin{proof}
    Any initial configuration of $k$-HWP covers exactly $(k-1)$ many edges of $G$. So, there are exactly $(n-1-(k-1)) = (n-k)$ unexplored edges, which need to be explored sequentially and require at least $(n-k)$ steps. 
    Consider all the edges $e$ with $d^k_{P_{Opt}}(e) = 1$. To prove the lemma, it is sufficient to prove that all such edges $e$ must be encountered at least twice by the algorithm $Opt$.
    Let $e = (u, v)$, where $u$ is the vertex nearer to $P_{Opt}$ than $v$.
    By Definition \ref{def:d_P}, there must exist a leaf $z$ such that the distance between $u$ and $z$ is at least $k$.
    Let the edge $e$ be encountered in some configuration $(\mathcal{C}_t)$ when an agent lies at $u$ and moves to $v$. Let $w$ be the vertex adjacent to $z$, and let $\mathcal{C}_{t'}$ be the configuration that encountered the edge $(w, z)$ for the first time. Then $t' > t$ must hold. After $\mathcal{C}_{t'}$, the edge $e$ is again encountered, as the algorithm $Opt$ terminates after exploring the vertex $y_{Sol}$. Thus, it must cross the edge $e$ while moving to $y_{Sol}$ after $\mathcal{C}_{t'}$.
    
    For example, consider Fig. \ref{fig:define-d-Tree}. Let $P_{Opt} = P_{xy}$ and $k = 3$. We have $d^k_{P_{Opt}}(e) = 1$. Let, the edges $e$ and $(u,v)$ encountered for the first time by the configurations $\mathcal{C}_t$ and $\mathcal{C}_{t'}$, respectively. Then the inequality $t'\geq t+2$ must hold and $\mathcal{C}_{t'}$ consists of the vertices $\{v, w, z\}$. The edge $e$ is encountered again after $\mathcal{C}_{t'}$, as the last encountered vertex is $y$. 
\qed \end{proof}

    \begin{lemma}
        \label{lem:optimal-k-HWP-Tree}
        For a tree $G$, $h_k^r(G) \geq (n-k) + \sum_{e \in E}d^k_{P^*}(e)$, where $P^*$ is the longest paths of the tree $G$ joining two vertices $v_1^*$ and $v_2^*$, which is selected by the Algorithm \textsc{k-RHWP-For-Tree}.
    \end{lemma}
    \begin{proof}
        Let $Opt$ be the optimal algorithm and $P_{xy}$ be the path between the vertices $x$ and $y$, where $x$ is the last explored vertex, and $x$ is the vertex in the initial configuration, which is farthest from $y$.
        Then $h_k^r(G) = t_{Opt} \geq (n-k) + \sum_{e \in E}d^k_{P_{xy}}(e)$, by Lemma \ref{lem:lower-Bound_k-Tree}.
        It is sufficient to prove that $\sum_{e \in E}d^k_{P_{xy}}(e) - \sum_{e \in E}d^k_{P^*}(e) \geq 0$.
        Since $G$ is a tree, there must exist exactly one sub-path that intersects both $P_{xy}$ and $P^*$, say $P_{uv}$. We assume that $len(P) - (k-1) \geq 0, \forall P \in \{P_{xu}, P_{v_1^*u}, P_{vy}, P_{vv_2^*}\}$. We set the corresponding value as zero if any of them is negative. All edges that lie on the path $P_{xu}$ must have the corresponding $d^k_{P_{xy}}(.)$ value is zero. However, the associated $d^k_{P^*}(.)$ may not be zero. Although there exist exactly $len(P_{xu}) - (k-1)$ number of edges on the path $P_{xu}$ with associated $d^k_{P^*}(.)$ value as one, as these are the edges $e$ for which the height of the rooted tree $T_{P_{OPT}}(e)$ is more than $(k-2)$. On the other hand, all the edges that lie on the path $P_{v_1^*u}$ must have corresponding $d^k_{P^*}(.)$ value is zero, but there exists an exact $len(P_{v_1^*u}) - (k-1)$ number of edges on the path $P_{v_1^*u}$ with associated $d^k_{P_{xy}}(.)$ value as one. Similar arguments are also followed when considering the path $P_{vy}$ and $P_{vv_2^*}$. Thus $\sum_{e \in E}d^k_{P_{OPT}}(e) - \sum_{e \in E}d^k_{P^*}(e) = (len(P_{v_1^*u}) - (k-1)) - (len(P_{xu} - (k-1)) + (len(P_{vv_2^*}) - (k-1)) - (len(P_{vy}) - (k-1)) = len(P_{v_1^*v_2^*}) - len(P_{xy}) \geq 0$, as $P^* = P_{v_1^*v_2^*}$ is the longest path among all the paths of $G$.
    \qed \end{proof}

    The above analysis shows that the algorithm \textsc{k-RHWP-For-Tree} returns a spanning walk of $G$ whose length matches the lower bound of $h_k^r(G)$. Therefore, combining Lemmas~\ref{lem:Correctness-k-HWP}, \ref{lem:time-k-RHWP-for-tree}, \ref{lem:lower-Bound_k-Tree}, and \ref{lem:optimal-k-HWP-Tree}, we establish the following theorem.

\begin{theorem}
\label{thm:OptimalBound_k-RHWP_Tree}
    Algorithm \textsc{k-RHWP-For-Tree} provides an optimal solution for k-RHWP in polynomial time, when the input graph $G$ is a tree.
\end{theorem}

\section{Algorithms for Arbitrary Graph in Case of k=2}   
\label{sec:2-HWP}
Here, we present an approximation algorithm for the 2-HWP on an arbitrary graph $G$. We start by introducing a simple algorithm that achieves an asymptotic approximation factor of 3, using the relationship between $h_1(G)$ and $h_2(G)$.

\begin{lemma}
    \label{lem:rel-h1-h2}
    $h_1(G) \leq 2 \cdot h_2(G)+1.$
\end{lemma}

\begin{proof}
    Let $W_2=\{(u_1, u_2), (u_3, u_4), \cdots, (u_{2p+1}, u_{2p+2})\}$ be an optimal spanning walk of configuration with two agents. Then $p = h_2(G)$. Also by the definition of adjacent configurations,$(u_{2i+1}, u_{2i+2} \in E, \forall i:0\leq i \leq p$; $u_{2i+1} = u_{2i+3}$ or $(u_{2i+1}, u_{2i+3}) \in E, \forall i:0\leq i \leq p-1$; $u_{2i+2} = u_{2i+4}$ or $(u_{2i+2}, u_{2i+4}) \in E, \forall i:0\leq i \leq p-1$. Hence, we can convert $W_2$ to a spanning walk $W_1$ of $G$, where $W_1=\{u_1, u_2, u_4, u_3, u_5, u_6, u_8, u_7, \cdots\}$. The length of $W_1$ is at most $2p+1$. Therefore, $h_1(G) \leq 2p+1 \leq 2h_2(G)+1.$
\qed \end{proof}

\noindent \textbf{\textit{The Approximation Algorithm for 2-HWP}:} Let $W_1 = \{v_1, v_2, \cdots, v_{l+1}\}$ be a spanning walk of configurations with one agent using some known $\beta$-approximation algorithm for the open Hamiltonian walk problem.
Then $l \leq \beta \cdot h_1(G)$. The walk $W_1$ can be converted to a spanning walk of configurations $W_2$ with two agents, where $W_2 = \{(v_1, v_2), (v_2, v_3), \cdots, (v_l, v_{l+1})\}$. The length of the walk $W_2$ is $(l-1) \leq \beta \cdot h_1(G)-1 \leq 2\beta \cdot h_2(G) + (\beta-1)$, by Lemma \ref{lem:rel-h1-h2}.
The best known approximation factor for the Hamiltonian walk problem is $3/2$ by Hoogeveen \cite{HOOGEVEEN1991291}. Hence, we have a solution of size at most $(3h_2(G) +1)$. 

Now, we want to find out if there is an approximation algorithm for 2-HWP that has an approximation factor less than 3.
Theorem \ref{thm:<3-2-HWP} conclusively confirms the possibility. 
The proposed algorithm uses two subroutines. One is \texttt{modified Christofide's heuristic} \cite{HOOGEVEEN1991291} to find the minimum-weight Hamiltonian path in a complete weighted graph. The other is an approximation algorithm for the \textit{maximum weighted $r$-set packing problem} ($max_rSPP$) \cite{thiery2023improvedweightedpacking}, where we are given a collection of weighted sets, each of which contains at most $r$ elements from a fixed universe of finite elements. The goal is to find a sub-collection of pairwise disjoint sets with the maximum total weight.
    
\noindent \textbf{\textit{Overview of the Improved Algorithm for 2-HWP:}} First, we create a weighted set-packing instance from the given unweighted graph $G$ (Section \ref{subsec:set-packing-instance}). Then, we apply an approximation algorithm of the set-packing problem to get a solution, and from that solution, we construct a tree $Tr(G)$ that spans all the vertices of $V$ (Section \ref{sec:constructing-tree}). Next, we find an open Hamiltonian walk in the tree $Tr(G)$, which gives us the required solution (Section \ref{subsec:find-HW}).
The proofs of all lemmas are excluded due to page limits.

\subsection{Creating a set-packing instance from $G$}
\label{subsec:set-packing-instance}
From a given instance of $2$-HWP, we first create an instance of $max_6SPP$, denoted as $\mathcal{I}_{SP}(G)$. 
The motivation behind the construction is to select all possible cases from which we can get a 2-transition edge. These transitions may arise from a cycle of length four $(C_4)$ or a grid of size $2 \times 3$ or any combination of the two. Finding such a disjoint collection of sets is preferable, as it reduces repeated visits to certain vertices. This can be achieved by applying the set-packing algorithm. However, it is possible that two $C_4$'s intersect at a single vertex (as shown in Fig. \ref{fig:packing-type}). In such cases, all vertices of two $C_4$'s can be traversed using 4 configurations. So, selecting only disjoint sets does not fully satisfy our objective.
We define several types of packing sets to adapt such subgraphs in our solution. 
    
We construct four types of weighted sets: two corresponding to each 4-length cycle $C = \{v_1, v_2, v_3, v_4\}$, and two corresponding to each $2 \times 3$ grid $D = \{v_1, v_2, v_3, v_4, v_5, v_6\}$ in $G$.
Note that $G$ may not contain any $C_4$. In that case, we have a simpler algorithm with a better approximation factor, which is discussed at the end of the section (Theorem \ref{thm:1-transit-edge_tree}). However, our main focus is the arbitrary case where the graph may be dense and contains many $C_4$'s.
\begin{itemize}
       \item (Type-I) A set $C_1^1 = C$ with weight four.
       \item (Type-II) Four sets, each of the form $C_2^i = C \setminus \{{v_i}\} \cup \{{u_C^i}\}$ with weight three for $1 \leq i \leq 4$, where $u_C^i$ is a dummy variable corresponding to $v_i$, not belonging to $V$ and not previously used in the construction process.
       \item (Type-III) A set $D = \{v_1, v_2, v_3, v_4, v_5, v_6\}$ with weight six.
       \item (Type-IV) For each set $D$ of Type-III, we create another six sets, each of the form $D \setminus \{{v_i}\} \cup \{{w_D^i}\}$ with weight five for $1 \leq i \leq 6$, where $w_D^i$ is a dummy variable corresponding to $v_i$, not belonging to $V$ and not previously used.
\end{itemize}
We call all these sets the \textit{packing set}. The weights are assigned according to how many vertices a packing set contains from the vertex set $V$. Note that $\{v_1, v_2, v_3, v_4\}$ being a cycle of length four, between the configurations $(v_1, v_2)$ and $(v_4, v_3)$, there is a transition edge. Similarly, if $\{v_1, v_2, v_3, v_4, v_5, v_6\}$ is a grid of size $2 \times 3$, there must be two transition edges between $(v_1, v_2)$ and $(v_4, v_3)$, and between the configurations $(v_4, v_3)$ and $(v_5, v_6)$.
\begin{figure}[h]
    \centering      \includegraphics[width=0.9\linewidth]{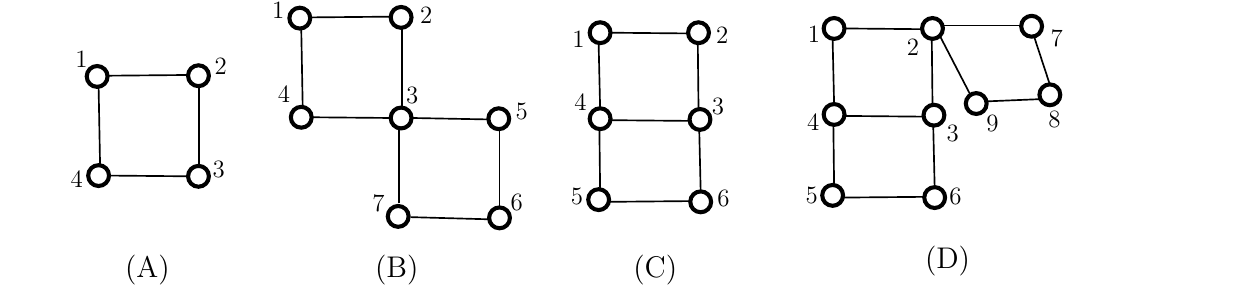}
    \caption{(A) A Type-I set $\{1, 2, 3, 4\}$, (B) A Type-I set $\{1, 2, 3, 4\}$ and a Type-II set $\{u, 5, 6, 7\}$, where $u$ is a dummy variable of the vertex $3$, (C) A Type-III set $\{1, 2, 3, 4, 5, 6\}$, (D) A Type-I set $\{2, 7, 8, 9\}$ and a Type-IV set $\{1, v, 3, 4, 5, 6\}$, where $v$ is the dummy variable of the vertex $2$. These packing sets are used to accommodate the above-mentioned subgraphs in our solution.}
    \label{fig:packing-type}
\end{figure} 
       
\noindent $G$ can have at most $\binom{n}{4}$ cycles of length four. For each such cycle, we create one Type-I packing set and four Type-II packing sets. Furthermore, $G$ has at most $\binom{n}{6}$ grids of size $2 \times 3$. For each such grid, we make one Type-III packing set and six Type-IV packing sets. Thus, for a given graph $G$, we create at most $5 \cdot \binom{n}{4} + 7 \cdot \binom{n}{6} = O(n^6)$ sets, where each set is of size at most six and the universal set of the procedure is $V \cup(\cup_{C \in G}\cup_{i=1}^4 u_C^i) \cup (\cup_{D \in G}\cup_{j=1}^6 w_D^j)$, where $C$ is a 4 length cycle and $D$ is a $2 \times 3$ grid in $G$. All $u_C^i$ and $w_D^j$ ($\forall i, j)$ are distinct.
    
$max_kSPP$ is NP-hard and since the introduction of the set packing problem, a significant research has been conducted in the past few decades for the weighted, unweighted, and online cases  \cite{arkin1998localPacking}, \cite{chandra2001greedyWeightedpacking},  \cite{emek2012onlinePacking}, \cite{sviridenko2013largePacking}, \cite{thiery2023improvedweightedpacking}. The best known polynomial-time approximation algorithm of the weighted $k$-set packing problem is given by Thiery and Ward \cite{thiery2023improvedweightedpacking}, achieving a factor slightly less than $\frac{2}{k+1}$, which asymptotically approaches $\frac{4}{2k+1}$ as $k$ grows. We define $max_6(\mathcal{I})$ to represent the problem $max_6SPP$ for the instance $\mathcal{I}$. We denote the optimal value of $\maxG$ by $w(OPT_{SP}(G))$. Then the following lemma holds:
\begin{lemma}
    \label{lem:lowerB_2-approx}
    $h_2(G) \geq (n - \frac{w(OPT_{SP}(G))}{2} - 1)$.
\end{lemma}
\begin{proof}
    Let $\mathcal{W}^*$ be the optimal open spanning walk for the $2$-HWP of graph $G$. Every transition edge in this walk must be a $1$-transition edge or a $2$-transition edge. Based on $\mathcal{W}^*$, we now construct an instance $\mathcal{I}_{SP}(\mathcal{W}^*)$ of $max_6SPP$ as follows: for each $2$-transition edge between $(v_1, v_2)$ and $(v_4, v_3)$, we always have a cycle $C = \{v_1, v_2, v_3, v_4\}$. For each such cycle $C$, we create a Type-I set $\{v_1, v_2, v_3, v_4\}$ with weight four, and four additional Type-II sets $\{u_C^1, v_2, v_3, v_4\},$ $ \{v_1, u_C^2, v_3, v_4\}, \{v_1, v_2, u_C^3, v_4\}, \{v_1, v_2, v_3, u_C^4\}$, each with weight three. For each two consecutive $2$-transition edges between $(v_1, v_2)$, $(v_4, v_3)$ and $(v_5, v_6)$, we create a Type-III set $D=\{v_1, v_2, v_3, v_4, v_5, v_6\}$ with weight six, along with six Type-IV sets of the form $D \setminus \{v_i\} \cup \{w_D^i\}$ ($1 \leq i \leq 6$), each with weight five. It is easy to observe that, $\mathcal{I}_{SP}(\mathcal{W}^*) \subseteq \mathcal{I}_{SP}(G)$. So, the optimal solution of $\max_6(\mathcal{I}_{SP}(\mathcal{W}^*))$ must be a feasible solution of $\maxG$. Therefore,
    \begin{equation}
    \centering
    \label{eq:rel_btw_OPT_SP_G_W*}
      w(OPT_{SP}(G)) \geq w(OPT_{SP}(\mathcal{W}^*))
    \end{equation}
    where $w(OPT_{SP}(\mathcal{W}^*))$ is the optimal solution of $\max_6(\mathcal{I}_{SP}(\mathcal{W}^*))$.
        
    The above construction shows that each of Type-I set (weight four) and Type-II set (weight three) requires two configurations, while each of Type-III (weight six) and Type-IV (weight five) requires three. Thus, the procedure uses at most a weight of two and at least $3/2$ per configuration. Therefore, the following holds:
    \begin{equation}
   \label{eq:bound_SPP}
       \frac{w(OPT_{SP}(\cW^*))}{2} \leq |\cW_1^*| \leq \frac{2 \cdot w(OPT_{SP}(\cW^*))}{3}
    \end{equation}
    where $\mathcal{W}_1^*$ denotes the set of all configurations $\cC \in \mathcal{W}^*$ such that there exists a transition edge in $\cW^*$ adjacent to $\cC$ and the corresponding packing set to that transition edge is in the optimal solution of $\max_6(\mathcal{I}_{SP}(\mathcal{W}^*))$. 
    Let $V' \subseteq V$ be the set of vertices that are in $\mathcal{W}_1^*$. Then the size of $V'$ must be $w(OPT_{SP}(\mathcal{W}^*))$. In other words, $|V \setminus V'| = n - w(OPT_{SP}(\cW^*))$. Let $\mathcal{W}_2^*$ be the set of all remaining configurations of $\mathcal{W}^*$ that are not in $\mathcal{W}_1^*$. We will prove that, on average, each configuration in $\mathcal{W}_2^*$ (except the initial one) covers at most one element of $V \setminus V'$.
        
    Let $\cC_t =(v_1^t, v_2^t)$ be an arbitrary configuration of $\cW_2^*$. We distinguish between two cases based on whether the previous configuration $\cC_{t-1}$, belongs to $\cW_1^*$ or not. (Case-1) $\cC_{t-1} \in \cW_1^*$. If the transition edge between $\cC_{t-1}$ and $\cC_{t}$ is a $1$-transition edge, we have done. Let the edge between $\cC_{t-1}$ and $\cC_{t}$ is a 2-transition edge and both the elements of $\cC_t$ are in $V \setminus V'$. In this case, the optimum value of $max_6SPP$, for the instance $\mathcal{I}_{SP}(\mathcal{W}^*)$, can be uplifted by adding $\cC_t$ as follows: If $\cC_{t-1}$ is a part of Type-I packing set, the set $C_1 = \{v^1_{t-2}, v^2_{t-2}, v^1_{t-1}, v^2_{t-1}\}$ must be a Type-I packing set of weight four and lies in the optimum solution. We can increased the optimal value by deleting $C_1$ and adding a Type-III packing set $\{v^1_{t-2}, v^2_{t-2}, v^1_{t-1}, v^2_{t-1}, v_t^1, v^2_t\}$ which is of weight six. This is a contradiction. Similarly, if $\cC_{t-1}$ is a part of the Type-II packing set, we can change the optimal value by removing the associated packing set of weight three and then inserting a Type-IV packing set, which is of weight five. If $\cC_{t-1}$ is a part of the Type-III set, we can change the optimal value by deleting the set of weight six and inserting two sets of Type-I with each of the weights as four. If $\cC_{t-1}$ is a part of the Type-IV set, we can change the optimal value by deleting the corresponding packing set of weight five and inserting a set of Type-I with weight four and a set of Type-III of weight three. All of these updates lead us to a contradiction with the notion of the optimality of $\cW^*$. Hence, if the transition edge between $C_t$ and $C_{t-1}$ is a $2$-transition edge, then at least one of the elements of $\cC_t$ must be in $V'$. (Case-2) $\cC_{t-1} \notin \cW_1^*$. For this case if the transition edge between $C_t$ and $C_{t-1}$ is a $1$-transition edge, then we have done. If the transition edge is $2$-transition edge and if we assume that there are at least three elements of $\cC_t \cup \cC_{t-1}$ which are not in $V'$, say $v_t^1, v_{t-1}^1, v_{t-1}^2$, then we can again increase the optimum value by adding the set $\{v_t^1, v_{t-1}^1, v_{t-1}^2, u\}$ with the weight three for some $u \notin V$, which is a contradiction. So, if the transition edge between $\cC_t$ and $\cC_{t-1}$ is $2$-transition edge, then at most two elements among the four elements of $\{\cC_t \cup \cC_{t-1}\}$  are in $V \setminus V'$. A similar argument is also valid if the initial configuration $\cC_0$ is in $\cW_2^*$. To prove this, we can consider the two cases based on $\cC_1$ is in $\cW_1^*$ or not.

    Therefore a configuration in $\cW_2^*$ covers, on average, at most, one vertex of $V \setminus V'$. Thus $|\cW_2^*| \geq |V \setminus V'| \geq n - w(OPT_{SP}(\cW^*))\geq n - w(OPT_{SP}(G))$.
    Hence by using Eq. \ref{eq:rel_btw_OPT_SP_G_W*} and  \ref{eq:bound_SPP} we have,
    \begin{equation}
    \begin{split}
       h_2(G) & = |\cW_1^*| + |\cW_2^*| -1 \geq \frac{w(OPT_{SP}(\cW^*))}{2} + n - w(OPT_{SP}(\cW^*)) - 1, \\
            & \geq n - \frac{w(OPT_{SP}(G))}{2} - 1.
    \end{split}
    \end{equation}
        \qed 
    \end{proof}

Our first step in designing an approximation algorithm is to create a tree $Tr(G)$ based on the approximate solution of $\maxG$. There are two types of vertices in the tree $Tr(G)$: some are the vertices of $G$, and the remaining ones are the edges of $G$, which we refer to as \textit{contracted vertices}. We discuss the construction process in the following section. 
\subsection{Construction of the tree Tr(G) that spans V}
\label{sec:constructing-tree}
Let $Sol$ be one of the approximation solutions of $\maxG$. Let $w(Sol_{SP}(G)$ $)$ be the value (weight) of the solution returned by  $Sol$. Now we do the following:
\begin{itemize}[left=0pt]
    \item For each Type-I packing set $\{v_1, v_2, v_3, v_4\} \in Sol$ of weight four, we add two contracted vertices $(v_1, v_2)$, $(v_4, v_3)$ to $\tL$ and connect them by an edge. Adding this edge spans 4 distinct vertices of $G$ and creates a new component.
    \item For each Type-II packing set $(C=\{v_1, v_2, v_3, v_4\}) \setminus \{v_i\} \cup \{u_C^i\} \in Sol$ of weight three for some $1 \leq i \leq 4$, we add two contracted vertices $(v_1, v_2)$ and $(v_4, v_3)$ to $\tL$, and connect them with an edge, which span three distinct vertices of $G$, since the vertex $v_i$ is already in another packing set $S \in Sol$. Let $v_j \in C$ and $v_l \in S$ be the associated vertices contracted with $v_i$. We then add an edge between $(v_i, v_j)$ and $(v_i, v_l)$, ensuring that all the vertices of $C$ and $S$ are part of the same connected component of $Tr(G)$.
    \item For each Type-III packing set $\{v_1, v_2, v_3, v_4, v_5, v_6\} \in Sol$ of weight six, we build three contracted vertices $(v_1, v_2)$, $(v_4, v_3), (v_5, v_6)$ and add to $\tL$. We join the edges between the first and last two, which span six distinct vertices of $V$ and create one new component.
    \item For each Type-IV packing set $D = (\{v_1, v_2, v_3, v_4, v_5, v_6\}) \setminus \{v_i\} \cup \{w_D^i\} \in Sol$ of weight five for some $1 \leq i \leq 6$, we select three contracted vertices $(v_1, v_2)$, $(v_4, v_3), (v_5, v_6)$ and add them to $\tL$ with the edges between the first and last two, which span five distinct vertices (except $v_i$) of $V$, where $v_i$ is present in another packing set $S \in Sol$. Let $v_j \in C$ and $v_l \in S$ be the associated vertices contracted with $v_i$. We also placed an edge between $(v_i, v_j)$ and $(v_i, v_l)$ that makes all the vertices of $C$ and $S$ in one component of $Tr(G)$.
\end{itemize}

The number of vertices in $V$ included by the above procedure is $w(Sol_{SP}(G))$. Let, $Sol = Sol(I) \cup Sol(II) \cup Sol(III) \cup Sol(IV)$, where $Sol(i)$ consists of all the packing sets of Type-i, $i \in \{I, II, III, IV\}$. So, the total number of edges added after the procedure is $\frac{w(Sol(I)_{{SP}}(G))}{4} + \frac{2 \cdot w(Sol(II)_{{SP}}(G))}{3} + \frac{w(Sol(III)_{{SP}}(G))}{3} + \frac{3 \cdot w(Sol(IV)_{{SP}}(G))}{5}$.
The number of trees in the forest $\tL$ (disjoint union of trees) after the procedure is $\frac{w(Sol(I)_{{SP}}(G))}{4} + \frac{w(Sol(III)_{{SP}}(G))}{6}$.
The number of vertices in $V$ yet to be spanned is $(n - w(Sol_{SP}(G)))$. Since the graph $G$ is connected, we can always find a contracted vertex $(v, v')$ in $Tr(G)$ such that $v \in V$ is not spanned by the above procedure. We add the vertex $v'$ to the forest $Tr(G)$ and join the edge between the vertex $v$ and the contracted vertex $(v, v')$. After this, we again select a vertex from $V$ that is not present in the forest $Tr(G)$ and connect it to one of the vertices of $Tr(G)$. We continue the procedure until it spans all the vertices of $V$.
The procedure includes $(n - w(Sol_{SP}(G)))$ additional edges to span all remaining $(n - w(Sol_{SP}(G)))$ vertices of $V$. Thereafter, we make the forest $\tL$ connected or a single tree.
To do so, we find two trees $Tr_1$ and $Tr_2$ of the forest $\tL$ such that there is an edge $e=(v_1, v_2) \in E$ such that $v_1$ is present in the tree $Tr_1$ and $v_2$ is present in the tree $Tr_2$. We include a contracted vertex $(v_1, v_2)$ in the forest $Tr(G)$. Then we connect an edge between $(v_1, v_2)$ and the vertex of $Tr_1$ which contains $v_1$, and another edge between $(v_1, v_2)$ and the vertex of $Tr_2$ which contains $v_2$.
Since there is a $n_c = \frac{w(Sol(I)_{{SP}}(G))}{4} + \frac{w(Sol(III)_{{SP}}(G))}{6}$ number of trees after the first phase of the procedure, the last phase of the procedure needs exactly $(2n_c-2)$ many edges. Hence, the total number of edges added after all the procedures is,
\begin{align}
\label{eq:bound_len_Tr(G)}
    len(\tL) &= \frac{w(Sol(I)_{SP}(G))}{4} + \frac{2\cdot w(Sol(II)_{SP}(G))}{3} + \frac{w(Sol(III)_{SP}(G))}{3} \notag \\
    &\quad + \frac{3 \cdot w(Sol(IV)_{SP}(G))}{5} + (n - w(Sol_{SP}(G))) \notag \\
    &\quad + 2\cdot (\frac{w(Sol(I)_{SP}(G))}{4} + \frac{w(Sol(III)_{SP}(G))}{6} - 1
\end{align}

\subsection{Finding an open Hamiltonian walk on the tree $\tL$}
\label{subsec:find-HW}
The above procedure implies that an open Hamiltonian walk of $\tL$ with one agent corresponds to an open walk of configurations with two agents that spans all the vertices in $V$. Theorem \ref{thm:OptimalBound_1HWP_Tree} guarantees that for a tree with diameter $d$ and $(n-1)$ edges, the length of the optimal open Hamiltonian walk is $2(n-1) - d$. 
To find a better approximation result, here we give an alternative way to find the open Hamiltonian walk of $\tL$ by using \texttt{modified Christofide's heuristic} proposed by Hoogeveen \cite{HOOGEVEEN1991291}. The heuristic applies to the complete weighted graph with weight function satisfying the triangular inequality. Therefore, we first convert our constructed tree $\tL$ into a complete weighted graph $Tr^*(G)$ as follows: each existing edge of $\tL$ is assigned a weight of one, while each missing edge between two vertices $tr_1$ and $tr_2$ is given a weight equal to the length of the shortest path between them in the line graph of $G$ or in G, depending on whether the vertex is contracted or not. If $tr_1$ or $tr_2$ is a non-contracted vertex (i.e., in $V$), we compute the length between them in $G$. However, if both $tr_1$ and $tr_2$ are contracted vertices, we calculate the shortest length in the line graph of $G$. By assigning weights this way, the triangular inequality is preserved.
Next, similar to the \texttt{modified Christofide's heuristic}, we first find the minimum spanning tree ($MST_{Tr^*}$) of the newly created graph $Tr^*(G)$, which gives the tree $Tr(G)$. Then we find the minimum matching ($MM_{Tr^*}$) between all the odd-degree vertices, except for two. The detail and analysis of this heuristic can be found in \cite{HOOGEVEEN1991291}.
Let $cost(MST_{Tr^*})$ and $cost(MM_{Tr^*})$ be the cost of the $MST_{Tr^*}$ and $MM_{Tr^*}$, respectively.
We keep two odd-degree vertices unmatched as we look for the path, not the cycle. By adding the matching edges, all the vertices in $Tr(G)$ become even-degree vertices except two, which is sufficient to get an Eulerian path (a path traversing each edge exactly once).
The length of this Eulerian path is $cost(MST_{Tr^*})+ cost(MM_{Tr^*}) = len(Tr(G)) + cost(MM_{Tr^*})$.
This path can be easily modified to obtain an open walk of a configuration of $G$ with two agents, without changing its length.
Hence, if the complete procedure as discussed above is referred to as $ALG_2$, then the length of the transition walk returned by $ALG_2$ is $t_{ALG_2} = len(Tr(G)) + cost(MM_{Tr^*})$.\\
\begin{figure}[H]
    \centering
    \includegraphics[width=0.86\linewidth]{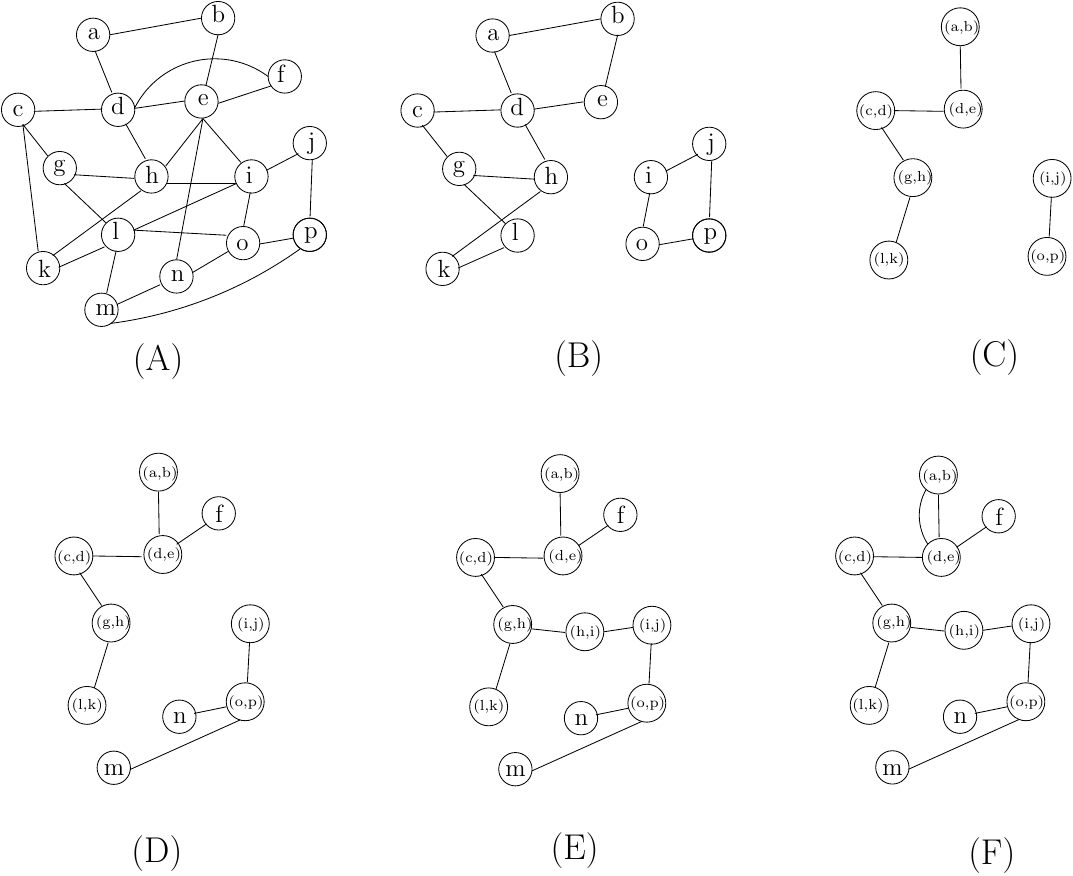}
    \caption{Illustrating a step-by-step execution of the algorithm} 
    \label{fig:algo-2-example}
\end{figure}
A step-by-step execution of the algorithm is shown in Fig. \ref{fig:algo-2-example}. 
Fig. \ref{fig:algo-2-example}(A) represents base graph $G$. Fig. \ref{fig:algo-2-example}(B) depicts all the packing sets $\{i,j,p,o\},$ $ \{a,b,e,u\},$ $ \{c,d,h,g,l,k\}$ selected by $Sol$, where $u$ is the dummy vertex for $d$.
Fig. \ref{fig:algo-2-example}(C) shows the creation of the forest $Tr(G)$. $\{i,j,p,o\}$ being a Type-I packing set, we incorporate two vertices $(i,j)$, $(o,p)$ within $Tr(G)$ and insert an edge between them. $\{c,d,h,g,k,l\}$ being a Type-III packing set, we add three vertices $(c,d)$, $(g, h), (l, k)$, and subsequently insert edges to connect them. $\{a,b,e,u\}$ being a Type-II packing set and $u$ being the dummy variable of $d$, we draw an edge between $(a,b)$ and $(d, e)$ and between $(c, d)$ and $(d, e)$.
Fig. \ref{fig:algo-2-example}(D) describes the process of connecting all the leftover vertices $\{f, m, n\}$ of $G$ to $Tr(G)$.
Fig. \ref{fig:algo-2-example}(E) portrays of converting the forest $Tr(G)$ into tree, where we select an edge $(h, i)$ that links them. Then, we insert the vertex $(h, i)$ and connect it to both $(g,h)$ and $(i,j)$.
In Fig. \ref{fig:algo-2-example}(F) we convert two odd-degree contracted vertices vertices $(a, b)$ and $(d,e)$, obtained from the same Type-II packing set, into even-degree by adding an edge between them.
A bound on $cost(MM_{Tr^*})$ is provided in the lemma below for the final analysis of the procedure $ALG_2$.

\begin{lemma}
    \label{lem:bound_on_minmatch}
    $cost(MM_{Tr^*}) \leq \frac{1}{2}(h_1(G)+ n\_odd_c)$, where $n\_odd_c$ is the number of odd-degree contracted vertices in the tree $Tr(G)$.
\end{lemma}
\begin{proof}
         Let $W$ be one of the optimal walks of configurations with one agent. Let $O = O_c \cup O_d$ be the set of odd-degree vertices in the tree $Tr(G)$, where $O_c$ is the set of all contracted odd-degree vertices and $O_d = O \setminus O_c$. We modify $W$ to make a walk for the tree $Tr(G)$ as follows: for each contracted vertex $(v_1, v_2) \in O_c$, we find the vertex $v_1$ when it first occurs in $W$, then we add the vertex $(v_1, v_2)$ in $W$ between the vertices $v_1$ and $u_1$, where $u_1$ is the next vertex of $v_1$ in the walk $W$. The length of the walk is $h_1(G)+n\_odd_c$. Then, similar to Theorem 1 \cite{arkin1998localPacking}, we can partition $W_1$ into two matchings, and then by selecting the minimum we can prove that $cost(MM_{Tr^*}) \leq \frac{1}{2}(h_1(G)+n\_odd_c)$.
     \qed \end{proof}

By using the above lemma, we have $t_{ALG_2} \leq len(Tr(G)) + \frac{1}{2}(h_1(G) + n_{odd_c})$.
To prove the following lemma, we modify the tree $Tr(G)$ as follows.
Let $tr_1$ and $tr_2$ be two contracted vertices corresponding to some Type-I or Type-II packing set. If both $tr_1$ and $tr_2$ are odd-degree vertices in the tree $Tr(G)$, then we add another edge between them. This modification increases the number of edges of $Tr(G)$ by one, while decreasing the number of odd-degree contracted vertices by two. Let $t'_{LAG_2}$ be the length of the transition walk returned by our algorithm $ALG_2$ after this modification. Then $t'_{ALG_2} \leq len(Tr(G)) - 1 + \frac{1}{2}(h_1(G) + n_{odd_c} -2) = len(Tr(G)) + \frac{1}{2}(h_1(G) + n_{odd_c})$. Thus, the upper bound on the length of the walk remains unaltered.
Therefore, we assume that, among the two contracted vertices constructed from a Type-I or Type-II packing set, at least one of them must be of even degree.
     
\begin{lemma}
    \label{lem:rel_nr_o_and_sol}
    $n\_odd_c \leq \frac{1}{4}w(Sol(I)_{SP}(G))+ \frac{1}{3} w(Sol(II)_{SP}(G)) + \frac{1}{3}w(Sol(III)_{SP}\\(G)) + \frac{3}{5} w(Sol(IV)_{SP}(G)))$, where $n_{odd_c}$ is the number of odd degree contracted vertices in the modified tree $Tr(G)$.
\end{lemma}

\noindent By applying Lemma \ref{lem:bound_on_minmatch} and \ref{lem:rel_nr_o_and_sol} and Eq. \ref{eq:bound_len_Tr(G)}, we get 
\begin{equation*}
    \begin{split}
        t_{ALG_2} &= len(Tr(G)) + cost(MM_{Tr^*})\\
        & \leq  n - w(Sol_{SP}(G)) + \frac{7\cdot w(Sol(I)_{{SP}}(G))}{8} + \frac{5 \cdot w(Sol(II)_{{SP}}(G))}{6} \\ 
        & \quad + \frac{11 \cdot w(Sol(III)_{{SP}}(G))}{12} + \frac{9 \cdot w(Sol(IV)_{{SP}}(G))}{10} - 2 + \frac{h_1(G)}{2}\\
        & = n - \frac{w(Sol_{SP}(G))}{12} -2 + \frac{h_1(G)}{2}
    \end{split}
\end{equation*}
We have $w(Sol_{SP}(G)) \geq \alpha \cdot w(OPT_{SP}(G))$, where $\alpha$ is the approximation factor of the approximation algorithm $Sol$ to the set packing problem $max_6SPP$. Furthermore, $w(OPT_{SP}(G)) \leq 2h_2(G) + 1$, as the maximum weight per configuration by the construction process is two. We also have $h_1(G) \leq 2h_2(G) +1$, from Lemma \ref{lem:rel-h1-h2} and $h_2(G) \geq (n - \frac{w(OPT_{SP}(G))}{2} - 1)$, from Lemma \ref{lem:lowerB_2-approx}. Combining all these inequalities, we have the following.
\begin{equation*}
    \begin{split}
        t_{ALG_2} &\leq  n - \frac{\alpha \cdot w(OPT_{SP}(G))}{12} -2 + h_2(G) -1\\
        & \leq  (n - \frac{w(OPT_{SP}(G)}{2} - 1) + (\frac{1}{2} - \frac{\alpha}{12})w(OPT_{SP}(G)) + h_2(G) - 2\\
        & \leq h_2(G) + (1 -\frac{\alpha}{6})h_2(G) + h_2(G) ~= (3-\frac{\alpha}{6})~h_2(G)
    \end{split}
\end{equation*}
     
\begin{theorem}
    \label{thm:<3-2-HWP}
    There exists a $(3-\frac{1}{21})$-approximation algorithm for $2$-HWP. If the graph has no $C_4$, the approximation factor reduces to two.
\end{theorem}

\begin{proof}
    The best known approximation algorithm to $max_6SPP$ is $\frac{2}{7}$ by Theiry and Hard \cite{thiery2023improvedweightedpacking}. Using their algorithm as $Sol$ and setting $\alpha = \frac{2}{7}$, we have $$t_{ALG_2} \leq (3- \frac{1}{21})h_2(G)$$
    
    \noindent If the graph has no $C_4$, the instance $\mathcal{I}_{SP}(G) = \emptyset$, implies that $w(OPT_{SP}(G))$ $ = 0$. Additionally, all the transitions are being $1$-transitions, $h_2(G) = h_2^r(G) \geq (n-2)$. 
    Hence, from the inequality above,  $t_{ALG_2} \leq (n - 2) + (h_2(G) - 1) \leq 2h_2(G)$.
   \qed 
\end{proof}
 
\section{k-HWP on Hypergraphs}

\label{sec:hypergraph}
Here, we formulate the problem $k$-HWP on hypergraphs and design an approximation algorithm to solve the problem.
The formulation of the problem and the proposed solution will provide insight to obtain a solution for $k$-HWP on the graph.
A hypergraph generalizes a graph by allowing each hyperedge to connect any number of vertices. In this paper, we consider the undirected $k$-uniform hypergraph without multiple edges, where each hyperedge connects exactly $k$ vertices. We misuse $G = (V, E)$ to represent both the graph and the hypergraph. However, a hyperedge in an $k$-uniform hypergraph is a subset of $V$ of size $k$. The line graph $L(G)$ of a hypergraph $G$ is a graph where each vertex represents a hyperedge of $G$, and two vertices are adjacent if their corresponding hyperedges intersect. The hypergraph $G$ is connected if $L(G)$ is connected.
    
Before defining the problem on the hypergraph, let us revisit the problem formulation on the simple graph with $k=2$. In $2$-HWP, we traverse the graph via the edges, with transitions classified as shifts or jumps.
If the transition edge between $(x_i, y_i)$ and $(x_j, y_j)$ is a 1-transition edge, the pairs share a common vertex, and the move is called a ``shift''. Otherwise, it's a 2-transition edge, where one agent moves along $(x_i, x_j)$ and the other along $(y_i, y_j)$. After their movements, they reach $(x_j, y_j)$ and the transition is called ``jump''. 
We introduce a similar notion for hypergraphs.
Let all $k$ agents initially lie on $k$ distinct vertices $\{v_1^0, v_2^0, \cdots, v_k^0\}$ of $G$ such that there is a hyperedge of $G$, say $e_0$, that connects all these vertices.
In the next step, the movements are one of the following. 
(i) All agents move to a hyperedge $e_1 \in E$ such that $e_0$ and $e_1$ are adjacent, i.e., they share at least one vertex.
This transition is named ``shift''.
(ii) $\forall i: 1\leq i \leq k$ the agent $a_i$ moves from $v_i^0$ to $v_i^1$ such that there exists a hyperedge that contains $\{v_i^0, v_i^1\}$.
Moreover, there exists a hyperedge that contains all the vertices $\{v_1^1, v_2^1, \cdots, v_k^1\}$.
This transition is named ``jump''.
Thus, we can assume each configuration to be one of the hyperedges, and the traversal is through those hyperedges.
We formally define the problem below. The NP-hardness of the problem is directly derived from the NP-hardness of 2-HWP, since the 2-uniform hypergraph is equivalent to a simple graph. An example of a hypergraph with its bipartite graph representation and one of the optimal solutions for $k$-HWP on Hypergraph is depicted in Fig. \ref{fig:subfigures}

\begin{definition}[Problem: $k$-HWP on Hypergraph]
    Consider an undirected connected $k$-uniform hypergraph $G$ along with a collection of $k$ agents. The task is to find a sequence of configurations $\{\mathcal{C}_i\}$ with minimum length where each successive pair of configurations is adjacent, and every vertex in $G$ is included at least once in one of these configurations.
\end{definition}

 \begin{figure}
	\centering
	\begin{subfigure}{0.31\linewidth}
		\includegraphics[width=\linewidth]{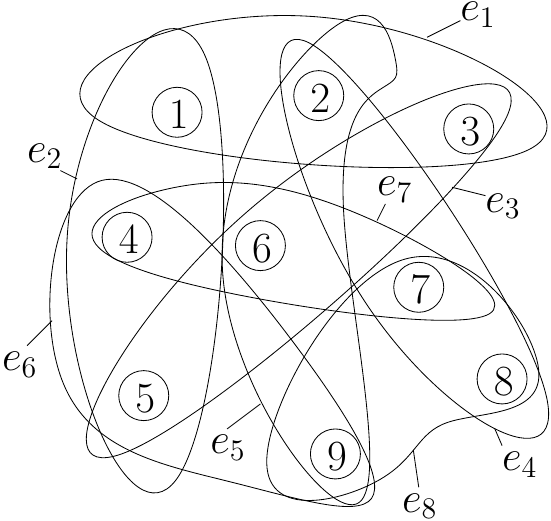}
		\caption{An expample of undirected  3-unform hypergraph $G$ with $V=\{1,2, \cdots, 9\}$, $E=\{e_1, e_2, \cdots, e_8\}$.}
		\label{fig:subfigA}
	\end{subfigure}
    \hspace{9pt}
	\begin{subfigure}{0.28\linewidth}
		\includegraphics[width=\linewidth]{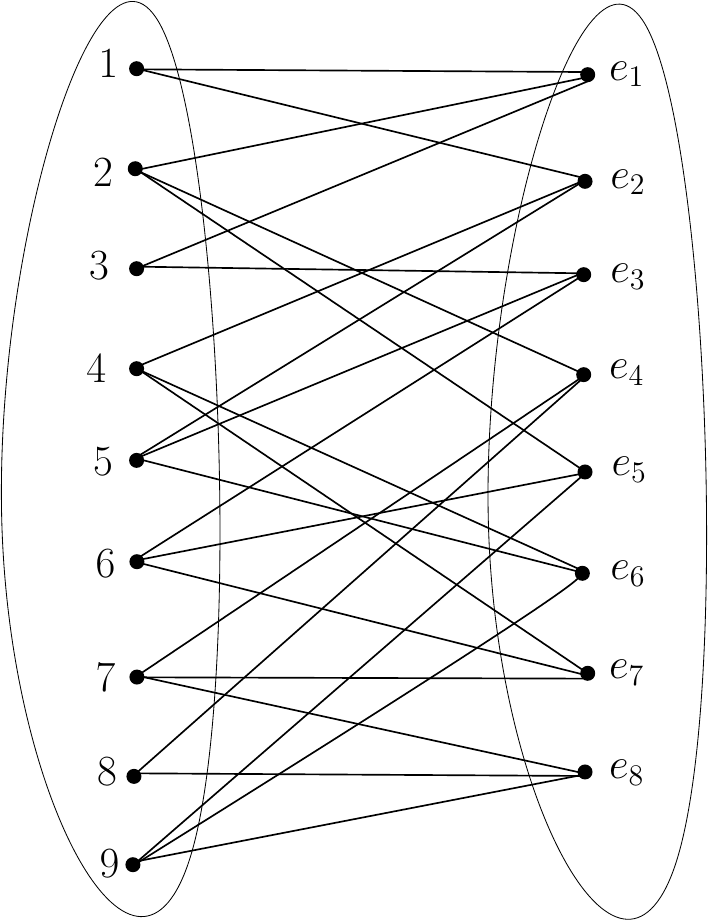}
		\caption{Bipartite graph representation of $G$, where the left partite set is the set $V$ and the right one is the set $E$.}
		\label{fig:subfigB}
	\end{subfigure}
    \hspace{9pt}
	\begin{subfigure}{0.31\linewidth}
	        \includegraphics[width=\linewidth]{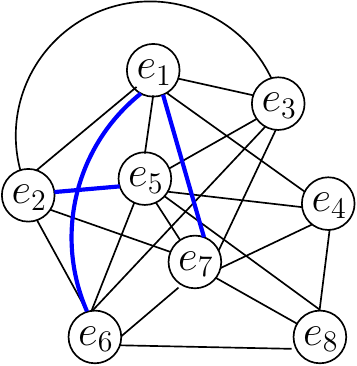}
	        \caption{The line graph $L(G)$ with some addition ``jump'' (blue heavy) edges between $e_1$ and $e_6$; $e_1$ and $e_7$; $e_2$ and $e_5$. }
	        \label{fig:subfigC}
         \end{subfigure}
	\caption{One of the optimal solution is the transition walk $\{e_1, e_7, e_6, e_8\}$.}
	\label{fig:subfigures}
\end{figure}

\noindent \textbf{High-level Idea of the Approx. Algorithm for k-HWP on Hypergraph:}
We begin by computing the line graph $L(G)$ of the given hypergraph $G$.
The graph $L(G)$ contains only the shift edges.
We then modify $L(G)$ by introducing all possible jump edges and denote the resulting graph by $L^*(G)$.
Next, we search for a minimum length walk of $L^*(G)$ that spans all the vertices of $V$.
To achieve this, we create an instance $\mathcal{I}_{CSC}(G)$ of the \textit{ minimum connected set cover problem} (MCSC) \cite{zhang2012connectedSetCover} from the graph $L^*(G)$.
In MCSC, we are given an universe of elements $U$, a family of subsets $\mathcal{F} = \{S\}$ where each $S \subseteq U$, and a connected graph where each vertex represents a unique $S \in \mathcal{F}$. 
The objective is to find a minimum size subfamily $\mathcal{F}' \subseteq \mathcal{F}$ such that each element of $U$ is present in some of the members of $\mathcal{F}'$ and the vertices associated with $\mathcal{F}'$ induce a connected subgraph of the given graph.  
For $\mathcal{I}_{CSC}(G)$, we consider the vertex set $V$ as the universe, the set of hyperedges $E$ as the family of subsets, and $L^*(G)$ is the underlying connected subgraph.
Using the best known approximation algorithm for MCSC by Zhang et al. \cite{zhang2012connectedSetCover}, which is of factor $(1 + \ln k)$, we derive a connected subgraph of $L^*(G)$, say $Sol$, that covers all the vertices of $V$.
Finally, we transform $Sol$ into a walk by doubling all its edges, yielding a solution of k-HWP on the hypergraph $G$.
The size of the solution is

\noindent \textbf{Detailed Description of the Approximation Algorithm for k-HWP on Hypergraph:}
\label{ref:subsec-app-descrip-algo-hypergraph}
First, we compute the line graph $L(G)$ of the given hypergraph $G$. The edges of $L(G)$ consist only of the shift edges. So, we remain to search all possible jump edges.
Therefore, we determine the open neighbourhood $N(v_i)$ of each vertex $v_i$, which is the set of vertices $u_i$ such that there exists an edge $e' \in E$ that contains $\{u_i, v_i\}$.
Now, for each edge $e=\{v_1,v_2, \cdots, v_k\} \in E$, we consider all possible sets of the form $\{u_1, u_2, \cdots, u_k\}$, where $u_i \in N(v_i)~ \forall i \in \{1, 2, \cdots, k\}$, and check if there is an edge $e''$ that connects the vertices $\{u_1, u_2, \cdots, u_k\}$.
If so, add an edge between $e$ and $e''$ on the line graph $L(G)$. There are at most $k\cdot \Delta(G)$ such sets for each $e \in E$, where $\Delta(G) = \max_{v\in V} |N(v)|$. We can discover all the jump edges in polynomial time. Let us give an unique id to each vertex $v_i \in V$ from the set $\{1, 2, \cdots, n\}$. We assume that the $k$-uniform hypergraph is stored in a $k$-dimensional array. The entry of $(i_1, i_2, \cdots, i_k\} \subset \{1, 2, \cdots, n\}$ is one if there is an edge that connects the vertices $\{v_1,v_2, \cdots, v_k\}$, otherwise zero. So, we can find whether there is an edge that connects the vertices $\{u_1, u_2, \cdots, u_k\}$ in $O(k)$ time. Thus, for a fixed hyperedge $e$, we can find all edges $e''$ for which there is a jump edge between them in $O(k^2 \cdot \Delta(G))$ time. Hence, all the jump edges of the line graph $L(G)$ can be computed in $O(m \cdot k^2 \cdot \Delta(G))$, where $m$ is the number of hyperedges of $G$. We denote this modified line graph by $L^*(G)$. Any solution of $k$-HWP on hypergraphs is an open walk on graph $L^*(G)$. Thus, the problem reduces to finding a minimum length walk of the graph $L^*(G)$ that covers all the vertices of $V$. We solve it using a subroutine of \textit{the minimum connected set cover problem}, defined below.

\noindent \textit{Minimum Connected Set Cover Problem(MCSC):} We have given a universe $U$ of elements and a family $\mathcal{F} = \{S\}$ of subsets of that universe such that $\cup S_{S \in \mathcal{F}} = V$. Any set cover is a subfamily $\mathcal{F}' \subseteq \mathcal{F}$ such that each element of $U$ is present in one of the elements in $\mathcal{F'}$. Minimum set cover is the set cover with the minimum cardinality.
The instance of the MCSC contains an additional connected graph whose vertex set is $\mathcal{F}$, which means that each $S \in \mathcal{F}$ is considered a vertex of the graph. The objective is to find a set cover $\mathcal{F}'$ with minimum cardinality such that the vertices associated with $\mathcal{F}'$ induce a connected subgraph of the given graph.  
Finding a minimum set cover is NP-hard. A simple greedy algorithm gives an approximation factor $(1+ \ln(\Delta))$, which is calculated using the pricing method and $\Delta = \max_{S \in \mathcal{F}}|S|$. Trevisan \cite{trevisan2001non} showed that the problem cannot be approximated within the factor $(\ln(\Delta)- O(\ln (\ln (\Delta))))$ unless P=NP. Dinur and Steurer \cite{dinur2014analytical} also established that the approximation of the set cover problem to within factor $(1-\epsilon)\ln n $ is NP hard for every $\epsilon>0$ and $n$ being the size of the universe. This result ensures that the $(1+ \ln(\Delta))$-approximation algorithm is asymptotically best unless P=NP. 
Zhang et al. \cite{zhang2012connectedSetCover} proposed the $(1+\ln(\Delta))$-approximation algorithm and an inapproximable bound to the MCSC.

From a given instance of $k$-HWP in hypergraph ($G$), we can construct an instance $\mathcal{I}_{CSC}(G)$ of the connected set cover problem by considering the vertex set $(V)$ as the universe, the set of hyperedges ($E$) as the family of subsets, and the line graph $L^*(G)$ as the underlying connected subgraph. Any solution to $k$-HWP on hypergraph $G$ is a feasible solution to the connected set cover problem for the instance $\mathcal{I}_{CSC}(G)$. So $OPT(\mathcal{I}_{CSC}(G)) \leq h_k(G)$, where $OPT(\mathcal{I}_{CSC}(G))$ is the optimum solution to the connected set cover problem for the instance $\mathcal{I}_{CSC}(G)$. Let $Sol$ be a solution to the connected set cover problem for the instance $\mathcal{I}_{CSC}(G)$ using a $\beta$-approximation algorithm from the literature, and let $|Sol|$ be the size of that solution. Then $|Sol| \leq \beta \cdot OPT(\mathcal{I}_{CSC}(G))$. For a given connected subgraph of $L^*(G)$ induced by the vertices of $Sol$, we can always construct a walk, by doubling all the edges and then finding an Eulerian walk. This walk is of length at most $2|Sol|$, which forms a feasible solution to $k$-HWP on the hypergraph $G$. So we have an approximate solution $(ALG_{HG})$ to $k$-HWP on $G$ with a solution size $(t_{ALG_{HG}})$ at most $2\cdot |Sol| \leq 2 \cdot \beta \cdot OPT(\mathcal{I}_{CSC}(G)) \leq 2\beta \cdot h_k(G)$.

\begin{theorem}
    \label{thm:approx-hypergraph}
    There exists a $2(1+\ln k)$-approximation algorithm for $k$-HWP on hypergraph.
\end{theorem}

\begin{proof}
    We obtain the factor by applying the best known approximation algorithm of the connected set cover problem, which is proposed by Zhang et al. \cite{zhang2012connectedSetCover}, i.e., by setting $\beta = (1+ \ln k)$, as in our case $\Delta = k$.
\qed \end{proof}

\begin{theorem}
   \label{thm:k-HWP-graph}
    There exists a $2(1+\ln k)$-approximation algorithm for $k$-HWP on a simple graph, when $k = O(1)$. 
\end{theorem}
\begin{proof}
    For a given connected graph $G$, we construct an instance of the connected set cover problem, where the universe is the vertex set $V$, and the family of subsets $(\mathcal{F})$ consists of all possible connected subgraphs of $G$ of size $k$. There can be at most $n^k$ such subsets of $V$. For a given $S \subset V$ of size $k$, we check the connectivity of the induced subgraph made by the vertices of $S$ in $O(k+m)$ time. Therefore, we need $n^{O(k)}$ time to construct the family $\mathcal{F}$. Next, we build the underlying connected graph $G$ in a manner similar to the above process mentioned in the algorithm for k-HWP on hypergraph. The only difference is that, for each vertex $v_i \in V$ we consider its closed neighbourhood $N[v_i] = N(v_i) \cup \{v_i\}$, as in the definition of $k$-HWP on the graph allows an agent to remain stationary. Thus, by a similar argument to that above, we can estimate the proposed result. However, the complexity of time is dominated by the polynomial $n^k$. So, the algorithm runs in polynomial-time algorithm only if $k = O(1)$.
\qed \end{proof}

\section{Conclusion}
\label{sec:conclusion}
In this paper, we examine the Hamiltonian walk problem within the context of a connected group of mobile agents that minimally traverse the graph.
The NP-hardness of the problem leads us to find some approximation algorithms. We propose an approximation algorithm for the problem with a set of two agents for any arbitrary graph. Then we define the problem on hypergraphs and provide an approximation algorithm.
We also show some optimal results for the acyclic graph. 
Additionally, we define a restricted version of the problem in which each transition explores only one new vertex compared to the previous one.
In the future, exploring tighter bounds for the problem and investigating the class of graphs in which a polynomial solution exists will be interesting.

\bibliographystyle{splncs04}
\bibliography{Algowin}

\end{document}